\documentclass{article}

\usepackage[affil-it]{authblk}
\usepackage[usenames,dvipsnames]{xcolor}
\usepackage{amsfonts}
\usepackage{amsmath,amsthm,amssymb,dsfont}
\usepackage{enumerate}
\usepackage[english]{babel}
\usepackage{graphicx}	
\usepackage[caption=false]{subfig}
\usepackage[margin=3cm]{geometry}
\usepackage{url}
\usepackage{todonotes}
\usepackage{bbm}
\usepackage{booktabs}
\usepackage{stmaryrd}
\usepackage{tikz}
\usetikzlibrary{chains}
\usetikzlibrary{fit}
\usepackage{pgflibraryarrows}		
\usepackage{pgflibrarysnakes}		
\usepackage{multirow}
\usepackage{makecell}
\usepackage{epsfig}
\usetikzlibrary{shapes.symbols,patterns} 
\usepackage{pgfplots}
\usepackage{bm}
\pgfplotsset{compat=newest}

\usepackage{hyperref}
\hypersetup{colorlinks=true,citecolor=blue,linkcolor=blue,filecolor=blue,urlcolor=blue,breaklinks=true}

\usepackage{nicefrac}
\usepackage{mathtools}
\usepackage{mathrsfs}
\usepackage{IEEEtrantools}

\theoremstyle{plain}
\newtheorem{theorem}{Theorem}[section]
\newtheorem{lemma}[theorem]{Lemma}

\newtheorem{proposition}[theorem]{Proposition}

\theoremstyle{definition}

\newtheorem{remark}[theorem]{Remark}

\newcommand\blfootnote[1]{%
	\begingroup
	\renewcommand\thefootnote{}\footnote{#1}%
	\addtocounter{footnote}{-1}%
	\endgroup
}

\newcommand*{\ra}{\rightarrow}

\newcommand*{\cA}{\mathcal{A}}
\newcommand*{\cB}{\mathcal{B}}
\newcommand*{\cC}{\mathcal{C}}

\newcommand*{\cE}{\mathcal{E}}

\newcommand*{\cH}{\mathcal{H}}
\newcommand*{\cL}{\mathcal{L}}
\newcommand*{\cI}{\mathcal{I}}

\newcommand*{\cN}{\mathcal{N}}
\newcommand*{\cM}{\mathcal{M}}

\newcommand*{\cT}{\mathcal{T}}

\newcommand*{\cV}{\mathcal{V}}

\newcommand*{\rDbf}{\mathrm{\mathbf{D}}}
\newcommand*{\rD}{\mathrm{D}}

\newcommand*{\RR}{\mathbb{R}}

\newcommand*{\CC}{\mathbb{C}}

\newcommand*{\NN}{\mathbb{N}}

\newcommand*{\GL}{\mathrm{GL}}

\newcommand*{\End}{\mathrm{End}}

\newcommand*{\id}{I}
\newcommand*{\poly}{\mathrm{poly}}

\newcommand*{\supp}{\mathrm{supp}}
\newcommand*{\tr}{\mathrm{tr}\,}
\newcommand*{\ket}[1]{| #1 \rangle}

\newcommand{\ketbra}[2]{|#1\rangle\!\langle #2|}

\newcommand*{\Pos}{\mathscr{P}}
\newcommand*{\Lin}{\mathscr{L}}

\newcommand*{\D}{\mathscr{D}}

\newcommand*{\CP}{\mathrm{CP}}

\newcommand{\CCap}{\mathrm{C}}
\newcommand{\QCapTW}{\mathrm{Q}^{\leftrightarrow}}
\newcommand{\QCapPPT}{\mathrm{Q}^{\mathrm{PPT},\leftrightarrow}}
\newcommand{\SQCapPPT}{\mathrm{Q}^{\mathrm{PPT},\leftrightarrow, \dagger}}
\DeclarePairedDelimiter\ceil{\lceil}{\rceil}

\newcommand*{\<}{\langle}
\renewcommand*{\>}{\rangle}

\newcommand {\br} [1] {\ensuremath{ \left( #1 \right) }}
\newcommand {\Br} [1] {\ensuremath{ \left[ #1 \right] }}
\newcommand {\cbr} [1] {\ensuremath{ \left\lbrace #1 \right\rbrace }}

\usepackage{scalerel,stackengine}
\newcommand\reallywidehat[1]{%
	\savestack{\tmpbox}{\stretchto{%
			\scaleto{%
				\scalerel*[\widthof{\ensuremath{#1}}]{\kern-.6pt\bigwedge\kern-.6pt}%
				{\rule[-\textheight/2]{1ex}{\textheight}}
			}{\textheight}%
		}{0.5ex}}%
	\stackon[1pt]{#1}{\tmpbox}%
}

\newcommand{\newD}{\rD^{\#}}
\newcommand{\newQ}{\mathrm{Q}^{\#}}

\newcommand{\reg}{\mathrm{reg}}

\newcommand{\suppress}[1]{}

\DeclareMathOperator*{\argmin}{argmin}

\begin{document}

\title{{\LARGE A hierarchy of efficient bounds on quantum capacities \\
 exploiting symmetry}}

\author[1]{Omar Fawzi}
\author[1]{Ala Shayeghi}
\author[1]{Hoang Ta}

\affil[1]{\small{Univ Lyon, ENS Lyon, UCBL, CNRS, Inria, LIP, F-69342, Lyon Cedex 07, France}}

\maketitle


\begin{abstract}
	Optimal rates for achieving an information processing task are often characterized in terms of regularized information measures. In many cases of quantum tasks, we do not know how to compute such quantities. Here, we exploit the symmetries in the recently introduced $\newD$ in order to obtain a hierarchy of semidefinite programming bounds on various regularized quantities. As applications, we give a general procedure to give efficient bounds on the regularized Umegaki channel divergence as well as the classical capacity and two-way assisted quantum capacity of quantum channels. In particular, we obtain slight improvements for the capacity of the amplitude damping channel. We also prove that for fixed input and output dimensions, the regularized sandwiched R\'enyi divergence between any two quantum channels can be approximated up to an $\epsilon$ accuracy in time that is polynomial in $1/\epsilon$.  

\end{abstract}


\section{Introduction}

The optimal rates for many quantum information processing tasks of interest can be characterized in terms of a regularized divergence between quantum channels. For a divergence $\mathbf{D}$ defined on quantum states, the corresponding channel divergence is defined by maximizing the divergence between the channel outputs over the set of possible inputs. There are two natural variants: for quantum channels $\cN$ and $\cM$ the non-stabilized divergence is given by only allowing input states $\rho$ in the input space of $\cN$ and $\cM$
\begin{align*}
	\overline{\mathbf{D}}( \cN \| \cM) = \sup_{\rho} \mathbf{D}(\cN(\rho) \| \cM(\rho) ) \ ,
\end{align*}
whereas the stabilized version allows arbitrary input states 
\begin{align*}
	\mathbf{D}( \cN \| \cM) = \sup_{\rho} \mathbf{D}((\cI \otimes \cN)(\rho) \| (\cI \otimes \cM)(\rho) ) \ ,
\end{align*}
where $\cI$ is the identity channel. The most well-known example illustrating these two variants is when $\mathbf{D}$ is the trace distance, then $\overline{\mathbf{D}}(\cN \| \cM)$ is the superoperator trace norm and $\mathbf{D}(\cN \| \cM)$ is the diamond norm, and it is known that we can have $\overline{\mathbf{D}}(\cN \| \cM) \ll \mathbf{D}(\cN \| \cM)$~\cite{kitaev2002classical}.

When analyzing tasks in the independent and identically distributed limit, an important divergence $\mathbf{D}$ is the Umegaki divergence $\rD$ defined by $\rD(\rho \| \sigma) = \tr(\rho \log \rho) - \tr(\rho \log \sigma)$. For example, in asymmetric hypothesis testing between channels $\cN$ and $\cM$, the Stein exponent is characterized using $\rD$, namely in terms of the \emph{regularized} channel divergence $\rD^{\mathrm{reg}}(\cN \| \cM) := \sup_{n} \frac{1}{n} \rD(\cN^{\otimes n} \| \cM^{\otimes n})$~\cite{wilde2020amortized,wang2019resource,fang2020chain}. In fact, it turns out that the channel divergence $\rD$ is in general non-additive~\cite{fang2020chain}, in which case we have $\rD^{\mathrm{reg}}(\cN \| \cM) > \rD(\cN \| \cM)$.
\blfootnote{ A preliminary version of this paper was presented at IEEE International Symposium on Information Theory, 2021}

Another example is the Holevo information of a quantum channel $\cN$, which is given by 
$\chi(\cN) = \min_{\sigma} \overline{\rD}(\cN \| \cT_{\sigma})$
where $\cT_{\sigma}$ is the replacer channel that outputs $\sigma$ for any input density operator~\cite{ohya1997capacities}.
The Holevo-Schumacher-Westmorland theorem (see e.g.,~\cite{wilde_2013}) states that the classical capacity of $\cN$ is given by $\chi^{\mathrm{reg}}(\cN) := \sup_n \frac{1}{n} \chi(\cN^{\otimes n})$ and the regularization is needed for some channels, as shown by~\cite{hastings2009superadditivity}. 

The objective of this paper is to provide efficient ways of computing, or more specifically upper bounding, such regularized channel divergences. In order to achieve this, we use the recently introduced $\rD^{\#}$ R\'enyi divergences~\cite{Omar_20} which were shown to give a converging hierarchy of upper bounds on regularized channel divergences. We exploit the symmetries of the resulting hierarchy of optimization programs to obtain a concise representation and solve it efficiently. Specifically, for quantum channels $\cN, \cM$, we show in Theorem~\ref{thm: computing_channel_divergence_1} that the permutation symmetry of the optimization program defining $\newD_{\alpha}(\cN^{\otimes k}\| \cM^{\otimes k})$ can be used to transform it into a semidefinite program with $\poly(k)$ variables and constraints compared to the straightforward representation which is of size exponential in $k$. However, as we will see, a direct implementation of this transformation would require an exponential time computation. In Theorem~\ref{thm: computing_channel_divergence}, we provide an algorithm which performs this transformation in $\poly(k)$ time, for fixed input and output dimensions. As a first application, we consider the task of approximating the regularized sandwiched R\'enyi divergence between two channels. Note that the sandwiched R\'enyi divergence (see Section~\ref{sec:prelim} for the definition) is in general non-additive~\cite{fang2020chain}, and it is not known whether its regularization is efficiently computable. Ref.~\cite{Fang_19} shows that the regularized quantity can be approximated up to arbitrary accuracy by $\frac{1}{k} \newD_{\alpha}(\cN^{\otimes k} \| \cM^{\otimes k})$, for sufficiently large $k$. Our results imply that the regularized sandwiched R\'enyi divergence between two channels can be approximated up to an accuracy $\epsilon\in(0,1]$, in time that is polynomial in $1/\epsilon$ (for fixed input/output dimensions).   
Furthermore, when the channels admit additional group symmetries, we present a general approach to combine these symmetries with the intrinsic permutation invariance to further simplify the problem. As an example demonstrating the potential of this approach, in Section~\ref{sec:Zsymmetry}, we apply our method to generalized amplitude damping channels and we show how a very simple symmetry of these channels leads to considerable reductions in the size of the convex optimization program for computing the channel divergence (see Table~\ref{table:Zsym}). 

In Section~\ref{section:5}, we present a procedure for efficiently computing improved strong converse bounds on the classical capacity of quantum channels by considering the generalized Upsilon-information~\cite{XWang_19} induced by the $\rD^{\#}$ R\'enyi divergences. To illustrate our method, we apply it to the amplitude damping channel (see Table~\ref{table:upperbound_capacity} for a comparison with the best previously known bounds). Even though the improvements we obtain for the classical capacity are very small for this channel, the amplitude damping channel $\cA_{p}$ is one of the current challenges as far as the classical capacity is concerned. In particular, it remains open whether $\chi^{\textrm{reg}}(\cA_{p}) = \chi(\cA_{p})$.
Finally, in Section~\ref{section:6}, we use our method for computing improved upper bounds on the two-way assisted quantum capacity of channels by considering the generalized Theta-information~\cite{Fang_19} induced by the $\rD^{\#}$ divergences and apply it to the amplitude damping channel, as an example (see Figure~\ref{fig:quantum_capacity} for a comparison with the best previously known bounds). 


\section{Preliminaries}
\label{sec:prelim}
\subsection*{Basic notation}

Let $\cH$ be a finite dimensional complex Hilbert space; we denote by $\Lin(\cH)$ the set of linear operators on $\cH$, $\Pos(\cH)$ denotes the set of positive semidefinite operators on $\cH$, and $\D(\cH):= \{\rho \in \Pos(\cH): \tr(\rho) = 1 \}$ is the set of density operators on $\cH$. For any two Hermitian operators $\rho, \sigma \in \Lin(\cH)$, we write $\rho \leq \sigma$ if $\sigma - \rho \in \Pos(\cH)$. Given $\rho \in \Lin(\cH)$, the support of $\rho$, denoted $\supp(\rho)$, is the orthogonal complement of its kernel. For $\rho,\sigma \in \Lin(\cH)$, we write $\rho \ll \sigma$, if $\supp(\rho) \subseteq \supp(\sigma)$. Let $X,Y$ be finite dimensional complex Hilbert spaces. For $A \in \Pos(X\otimes Y)$, we often explicitly indicate the quantum systems as a subscript by writing $A_{XY}$. The marginal on the subsystem $X$ is denoted $A_X =\tr_{Y}(A_{XY})$. Let $\cbr{\ket{x}}_x$ and $\cbr{\ket{y}}_y$ be the standard bases for $X$ and $Y$, respectively. We will use a correspondence between linear operator in $\Lin(Y,X)$ and vectors in $X\otimes Y$, given by the linear map $\mathrm{vec}:\Lin(Y,X)\rightarrow X\otimes Y$, defined as $\mathrm{vec}\br{\ketbra{x}{y}}=\ket{x}\ket{y}$.
\par
We denote by $\CP(X:Y)$ the set of completely positive (CP) maps from $\Lin(X)$ to $\Lin(Y)$. A \emph{quantum channel} $\cN_{X \to Y}$ is a CP and trace-preserving linear map from $\Lin(X)$ to $\Lin(Y)$. A \emph{subchannel} $\cM_{X \to Y}$ is a CP and trace-nonincreasing linear map from $\Lin(X)$ to $\Lin(Y)$.
Let $X'$ be isomorphic to $X$ and $\ket{\Phi}_{XX'} = \sum_{x}\ket{x}_{X} \ket{x}_{X'}$ be the unnormalized maximally entangled state. For a linear map $\cN_{X' \to Y}$, we denote by $J_{XY}^{\cN} \in \Pos(X \otimes Y)$ the corresponding \emph{Choi matrix} defined as $J_{XY}^{\cN} = (\cI_X \otimes \cN)({\ketbra{\Phi}{\Phi}}_{XX'})$, where $\cI_X$ denotes the identity map on $\Lin(X)$.


\textbf{Polynomial on a vector space.} For a finite dimensional complex vector space $\cH$, the \emph{dual vector space} $\cH^{*}$ of $\cH$ is the vector space of all linear transformations $\varphi: \cH \to \CC$. The \emph{coordinate ring} of $\cH$, denoted $\mathcal{O}(\cH)$, is the algebra consisting of all $\CC$-linear combinations of products of elements from $\cH^{*}$. An element of $\mathcal{O}(\cH)$ is called a polynomial on $\cH$. A polynomial $p \in \mathcal{O}(\cH)$ is called \emph{homogeneous} if it is a $\CC$-linear combination of a product of $k$ non-constant elements of $\cH^{*}$ (for a fixed non-negative integer $k$). We denote by $\mathcal{O}_{k}(\cH)$ the set all homogeneous polynomials of degree $k$. 

\subsection*{Quantum divergences}

A functional $\mathbf{D}: \D(\cH) \times \Pos(\cH) \to \RR$ is a \emph{generalized quantum divergence}~\cite{Polyanskiy2010,Sharma2012} if it satisfies the data-processing inequality
\begin{align*}
\mathbf{D}(\cN(\rho)\| \cN(\sigma)) \leq \mathbf{D}(\rho \| \sigma). 
\end{align*}

Let  $\rho \in \D(\cH)$ and $\sigma \in \Pos(\cH)$ such that $\rho \ll \sigma$. The \emph{sandwiched R\'enyi divergence}~\cite{sandwiched_diver2013},~\cite{sandwiched_diver2014} of order $\alpha \in (1,\infty)$ is defined as
\begin{align*}
\widetilde{\rD}_{\alpha}(\rho \| \sigma):=\frac{1}{\alpha-1} \log \tr\left[\left(\sigma^{\frac{1-\alpha}{2 \alpha}} \rho \sigma^{\frac{1-\alpha}{2 \alpha}}\right)^{\alpha}\right].
\end{align*}
The \emph{geometric R\'enyi divergence}~\cite{Petz1998ContractionOG,geometric_diver2013,tomamichel2015quantum,hiai2017, Fang_19} of order $\alpha$ is defined as
\begin{align*}
\widehat{\rD}_{\alpha}(\rho \| \sigma):=\frac{1}{\alpha-1} \log \tr\left[\sigma^{1 / 2}\left(\sigma^{-1 / 2} \rho \sigma^{-1 / 2}\right)^{\alpha} \sigma^{1 / 2}\right].
\end{align*}
The \emph{max divergence} is defined as
\begin{align*}
\rD_{\max}(\rho\| \sigma): = \log \inf \{\lambda >0:\rho \leq \lambda\sigma \}.
\end{align*}

The inverses in these formulations are generalized inverses, i.e., the inverse on the support. When $\rho \ll \sigma$ does not hold, these quantities are set to $\infty$. Recently, in \cite{Omar_20}, the authors introduced an interesting quantum R\'enyi divergence called \emph{$\#$-R\'enyi divergence}. To define this divergence, we recall the geometric mean of two positive definite matrices.  

For $\alpha \in (0,1)$, the $\alpha$-geometric mean of two positive definite matrices $\rho$ and $\sigma$ is defined as
\begin{align*}
\rho \#_{\alpha} \sigma = \rho^{1/2}(\rho^{-1/2}\sigma\rho^{-1/2})^{\alpha}\rho^{1/2}. 
\end{align*}
The $\alpha$-geometric mean has the following properties (see Refs. \cite{Kubo1980MeansOP} and \cite{Omar_20}):
\begin{enumerate}
	\item Monotonicity: $A\leq C$ and $B \leq D$ implies $A\#_{\alpha}B \leq C\#_{\alpha}D$.
	\item{ Transformer inequality: $M(A\#_{\alpha}B)M^{*} \leq (MAM^{*})\#_{\alpha}(MBM^{*})$, with equality if $M$ is invertible. 
	}
	\item $(aA)\#_{\alpha}(bB) = a(b/a)^{\alpha}(A \#_{\alpha})B$, for any $a>0$ and $b\geq 0$.
	\item {Joint-concavity/sub-additivity: for any $A_i,B_i \geq 0$ we have
		\begin{align*}
		\sum_{i} A_{i} \#_{\alpha}B_i \leq \left(\sum_{i}A_i \right)\#_{\alpha}\left(\sum_{i}B_i \right)\enspace.
		\end{align*}
	}
	\item {Direct sum: for any $A_1,A_2,B_1,B_2 \geq 0$, we have
		\begin{align*}
		(A_1 \oplus A_2) \#_{\alpha} (B_1 \oplus B_2) = (A_1 \#_{\alpha} B_1)\oplus (A_2 \#_{\alpha} B_2)\enspace,
		\end{align*}
		where $A_1 \oplus A_2$ form a block diagonal matrix $\left[\begin{array}{cc}
		A_{1} & 0 \\
		0 & A_{2}
		\end{array}\right].$
	}
\end{enumerate}

The \emph{$\#$-R\'enyi divergence}~\cite{Omar_20} of order $\alpha$ between two positive semidefinite operators  is defined as
\begin{align*}
\newD_{\alpha}(\rho \| \sigma)&:= \frac{1}{\alpha -1 }\log \newQ_{\alpha}(\rho \| \sigma), \\
\newQ_{\alpha}(\rho \| \sigma)&:= \min_{A \geq 0} \tr(A) \, \, \text{ s.t. } \, \, \rho \leq \sigma \#_{\frac{1}{\alpha}} A.   
\end{align*}
We note that the above convex program may be expressed as a semidefinite program when $\alpha$ is a rational number~\cite{fawzi2017lieb's,SAGNOL}. The order between these divergences is summarized in the proposition below. 
\begin{proposition}
	\label{prop:compair_divergences} 
	For any $\rho,\sigma \in \Pos(\cH)$ and $\alpha \in (1,2]$, we have
	\begin{align*}
	\rD(\rho \| \sigma) \leq \widetilde{\rD}_{\alpha}(\rho \| \sigma) \leq \newD_{\alpha}(\rho \| \sigma) \leq \widehat{\rD}_{\alpha}(\rho \| \sigma) \leq \rD_{\max}(\rho \| \sigma). 
	\end{align*}
\end{proposition}

For a quantum channel $\cN_{X' \to Y}$, a subchannel $\cM_{X' \to Y}$ and a generalized quantum divergence $\rDbf$ the corresponding channel divergence~\cite{Leditzky2018} is defined as
\begin{align*}
\rDbf(\cN \| \cM):= \sup_{\rho_X \in \D(X)} \rDbf(\cN_{X' \to Y}(\phi_{XX'})\| \cM_{X' \to Y}(\phi_{XX'})) \enspace,
\end{align*}
where $\phi_{XX'}$ is a purification of $\rho_X$. For $\rDbf = \newD_{\alpha}$, the channel divergence can be expressed in terms of a convex optimization program~\cite{Omar_20} as follows.
\begin{align}
\newD_{\alpha}(\cN \| \cM) &\coloneqq \frac{1}{\alpha-1} \log \newQ_{\alpha}(\cN \| \cM), \label{eq:sdp Dsharp 1}\\
\newQ_{\alpha}(\cN \| \cM) 
&\coloneqq \min_{A_{XY} \geq 0} \| \tr_{Y}(A_{XY}) \|_{\infty}  \quad \textup{ s.t. } \quad  (J^{\cN}_{XY}) \leq (J^{\cM}_{XY}) \#_{1/\alpha} A_{XY} \ , \label{eq:sdp Dsharp 2}
\end{align}
where $\| . \|_{\infty}$ denotes the operator norm. 

The generalization of $\newD_{\alpha}$ to channels is subadditive under tensor products~\cite{Omar_20}: For any $\alpha \in (1,\infty)$, quantum channels $\cN_1,\cN_2$, and subchannels $\cM_1,\cM_2$, we have
\begin{align*}
\newD_\alpha (\cN_1 \otimes \cN_2\|\cM_1 \otimes \cM_2) \leq \newD_\alpha (\cN_1 \|\cM_1) + \newD_\alpha (\cN_2 \|\cM_2) \, .
\end{align*}


\section{Tools for efficiently representing structured convex programs} \label{sec:Tools}

In this section, we provide the necessary mathematical background on how symmetries in a convex optimization problem can be utilized to represent the program efficiently, we refer the interested reader to references such as~\cite{Bachoc2012} for more information.  

\subsection{Matrix $*$-algebra background}
A subset $\cA$ of the set of all $n\times n$ complex matrices is said to be a \emph{matrix $*$-algebra} over $\CC$, if it contains the identity operator and is closed under addition, scalar multiplication, matrix multiplication, and taking the conjugate transpose. For our applications, the structure in the optimization programs we consider will allow us to assume that the variables live in such an algebra. A map $\varphi:\cA \to \cB$ between two matrix $*$-algebras $\cA$ and $\cB$ is called a $*$-isomorphism if
\begin{itemize}
	\item $\varphi$ is a linear bijection,
	\item $\varphi(AB) = \varphi(A)\varphi(B)$ for all $A,B \in \cA$,
	\item $\varphi(A^{*}) = \varphi(A)^{*}$ for all $A \in \cA$.
\end{itemize}
The matrix algebras $\cA$ and $\cB$ are called \emph{isomorphic} and we write $\cA \cong \cB$. Note that, by the second property above, $*$-isomorphisms preserve positive semidefiniteness. From a standard  result in the theory of matrix $*$-algebra, we get the following structure theorem. 
\begin{theorem} [Theorem 1,\cite{Gijswijt_thesis}] \label{thm:*-isomorphism decomposition}
	Let $\cA\subseteq \CC^{n\times n}$ be a matrix $*$-algebra. There are numbers $t$, $m_1, \ldots, m_t$ such that there is a $*$-isomorphism $\phi$ between $\cA$ and a direct sum of complete matrix algebras
	\begin{align}
		\phi \; : \; \cA \rightarrow \bigoplus_{i=1}^{t} \CC^{m_i \times m_i}\enspace.
	\end{align}
\end{theorem}

In other words, under the mapping $\phi$, all the elements of $\cA$ have a common block-diagonal structure. Moreover, this is the finest such decomposition for a generic element of $\cA$. We remark that the $*$-isomorphism $\phi$ can be computed in polynomial time in the dimension of the matrix $*$-algebra $\cA$ (see e.g., Theorem~$2.7$ in Ref.~\cite{Bachoc2012} and the following discussion, or  Ref.~\cite{Gijswijt_thesis}).

\subsubsection*{Regular $*$-representation}
In general, computing the block-diagonal decomposition above and the corresponding mapping is a non-trivial procedure. In this section, we introduce a simpler $*$-isomorphism which embeds $\cA$ into $\CC^{m\times m}$, where $m=\dim \cA$.

Let $\cA$ be a matrix $*$-algebra of dimension $m$ and $\cC=\{C_1,\ldots,C_m\}$ be an orthonormal basis for $\cA$ with respect to the Hilbert-Schmidt inner product. Let $L$ be the linear map defined for every $A\in \cA$ by the left-multiplication by $A$. Consider the matrix representation of $L$ with respect to the orthonormal basis $\cC$. For every $A\in \cA$,  $L(A)$ is represented by an $m\times m$ complex matrix given by $L(A)_{ij}=\langle C_i,AC_j \rangle$, for every $i,j\in[m]$. The map $L \,:\, \cA \rightarrow \CC^{m\times m}$, is called the \emph{regular $*$-representation} of $\cA$ associated with the orthonormal basis $\cC$. 
Since $L$ is a linear map, it is completely specified by its image for the elements of the basis $\cC$. Let $(p_{rs}^{t})_{r,s,t\in [m]}$ be the \emph{multiplication parameters} of $\cA$ with respect to the basis $\cC$ defined by $C_rC_s = \sum_{t=1}^{m} p_{rs}^{t} C_t$. Then, $L(C_r)_{ij}=p_{rj}^i$, for every $r\in[m]$.
\begin{theorem}[\cite{Klerk_07}] \label{thm: regular_rep}
	Let $\cL$ be the matrix $*$-algebra generated by the matrices $L(C_1),\dots,L(C_m)$. Then the map $\psi$ defined as
	\begin{align}
		\psi: \cA \rightarrow \cL \;,\; \psi(C_r) = L(C_r) \;,\; r \in[m]\enspace, \label{eq:regular rep}
	\end{align}
	is a $*$-isomorphism.
\end{theorem}
Note that under the $*$-isomorphism $\psi$ of Theorem~\ref{thm: regular_rep}, for $\cA \subseteq \CC^{n\times n}$, the matrix dimensions are reduced from $n\times n$ to $m\times m$, whereas the $*$-isomorphism $\phi$ of Theorem~\ref{thm:*-isomorphism decomposition} provides a fine block-diagonal decomposition into $t$ blocks where the block matrix $i$ is of size $m_i\times m_i$, satisfying $m = m_1^2+\ldots+m_t^2$.

\subsection{Representation theory}

  We recall some basic facts in representation theory of finite groups. For further details, we refer the reader to Refs.~\cite{rep_1977} and~\cite{Fulton1991}. Let $\cH$ be a finite dimensional complex Hilbert space and $G$ be a finite group. A \emph{linear representation} of $G$ on $\cH$ is a group homomorphism $\varrho: G \to \GL(\cH)$, where $\GL(\cH)$ is the general linear group on $\cH$. The space $\cH$ is called a \emph{$G$-module}. For $v \in \cH$ and $g\in G$, we write $g\cdot v$ as shorthand for $\varrho(g)v$. For $X\in\Lin(\cH)$, the action of $g\in G$ on $X$ is given by $\varrho(g)X\varrho(g)^*$.
     
   A representation $\varrho: G\to \GL(\cH)$ of $G$ is called \emph{irreducible} if it contains no proper submodule $\cH'$ of $\cH$ such that $g\cH' \subseteq \cH'$. Let $\cH$ and $\cH'$ be $G$-modules, a linear map $\psi: \cH \to \cH'$  is called a \emph{$G$-equivariant map}  if $g \cdot \psi(v) = \psi(g \cdot v)$ for all $g \in G, v\in \cH$.  Two $G$-modules $\cH$ and $\cH'$ are called \emph{$G$-isomorphic}, write $\cH \cong \cH'$, if there is a bijective equivariant map from $\cH$ to $\cH'$. We denote by $\End^{G}(\cH)$, the set of all $G$-equivariant maps from $\cH$ to $\cH$, i.e.,  
\begin{align*} 
	\End^{G}(\cH) = \{T \in \Lin(\cH): T(g\cdot v) = g\cdot T(v), \forall v \in \cH, g \in G \}.  
\end{align*}

Let $G$ be a finite group acting on a finite dimensional complex vector space $\cH$. Then the space $\cH$ can be decomposed as $\cH = \cH_1 \oplus \dots \oplus \cH_t$ such that each $\cH_i$ is a direct sum $\cH_{i,1} \oplus \dots \oplus \cH_{i,m_i}$ of irreducible $G$-modules with the property that $\cH_{i,j} \cong \cH_{i',j'}$ if and only if $i=i'$. The $G$-modules $\cH_1,\dots,\cH_t$ are called the \emph{$G$-isotypical components} and $(m_1,\dots,m_t)$ are called the \emph{multiplicities} of the corresponding irreducible representations.
\par 
It is straightforward to see that $\End^{G}(\cH)$ corresponds to the subset of $G$-invariant matrices and has the structure of a matrix $*$-algebra. For $\cA=\End^{G}(\cH)$, the structural parameters of Theorem~\ref{thm:*-isomorphism decomposition} have a representation theoretic interpretation. In particular, the number of the direct summands $t$ corresponds to the number of isomorphism classes of irreducible $G$-submodules and $m_i$ is the multiplicity of the irreducible $G$-submodules in class $i$. 
\par  
For each $i \in [t]$ and $j \in [m_i]$, let $u_{i,j} \in \cH_{i,j}$ be a nonzero vector such that for each $i$ and all $j,j' \in [m_i]$, there is a bijective $G$-equivariant map from $\cH_{i,j}$ to $\cH_{i,j'}$ that maps $u_{i,j}$ to $u_{i,j'}$. For $i \in [t]$, we define a matrix $U_i$ as $[u_{i,1},\dots,u_{i,m_i}]$, with $u_{i,j}$ forming the $j$-th column of $U_i$. 
The matrix set $\{U_1,\dots,U_t\}$ obtained in this way is called a \emph{representative} for the action of $G$ on $\cH$. The columns of the matrices $U_i$ can be viewed as elements of the dual space $\cH^{*}$ (by taking the standard inner product). Then each $U_i$ is an ordered set of linear functions on $\cH$.

Since $\cH_{i,j}$ is the linear space spanned by $G \cdot u_{i,j}$ (for each $i,j$), we have
\begin{align*}
	\cH = \bigoplus_{i=1}^{t}\bigoplus_{j=1}^{m_i} \CC G \cdot u_{i,j} \, ,
\end{align*}
 where $\CC G=\cbr{\sum_{g\in G}\alpha_g g: \alpha_g\in \CC}$ denotes the complex group algebra of $G$. Moreover, note that
 \begin{align} \label{eq:dim_END vs multiplicites}
 	\dim \End^{G}(\cH) = \dim \End^{G}\left(  \bigoplus_{i=1}^{t}\bigoplus_{j=1}^{m_i} \cH_{i,j}  \right) = \sum_{i=1}^{t} m_{i}^2 \, .
 \end{align}

Note that with the action of the finite group $G$ on the space $\cH$, any inner product $\< \, , \>$ on $\cH$ gives rise to a $G$-invariant inner product $\< \, ,\>_{G}$ on $\cH$ via the rule $\<x,y\>_{G} \coloneqq \frac{1}{|G|}\sum_{g \in G}\<g \cdot x,g\cdot y\>$. Let $\<\, ,\>$ be a $G$-invariant inner product on $\cH$ and $\{U_1,\dots,U_t\}$ be a representative for the action of $G$ on $\cH$. Consider the linear map $\phi : \End^{G}(\cH) \to \bigoplus_{i=1}^{t}\CC^{m_i \times m_i}$ defined as
\begin{align}
	\label{eq:block_diagonal_matrix}
	\phi(A) \coloneqq \bigoplus_{i=1}^{t} \left( \<Au_{i,j'},u_{i,j}\>\right)_{j,j'=1}^{m_i} \;,\; \forall A\in \End^{G}(\cH) \enspace.
\end{align} 
For $i \in [t]$ and $A\in \End^{G}(\cH)$, we denote the matrix $\left( \<Au_{i,j'},u_{i,j}\>\right)_{j,j'=1}^{m_i}$ corresponding to the $i$-th block of $\phi(A)$ by $\llbracket \phi(A) \rrbracket_i$. 
\begin{lemma}[Proposition 2.4.4, \cite{Polak_thesis}]
	\label{lemma:psd_preserving}
	The linear map $\phi$ of Eq.~\eqref{eq:block_diagonal_matrix} is bijective and for every $A \in \End^{G}(\cH)$, we have $A \geq 0$ if and only if $\phi(A) \geq 0$. Moreover, there is a unitary matrix $U$ such that
	\begin{align*}
			U^{*}AU = \bigoplus_{i=1}^{t} \bigoplus_{j=1}^{d_i} \llbracket \phi(A) \rrbracket_i \;,\; \forall A\in \End^{G}(\cH) \enspace ,
	\end{align*}
	where $d_{i} = \dim(\cH_{i,1})$, for every $i \in [t]$.
\end{lemma}

Lemma~\ref{lemma:psd_preserving} plays a very important role in our symmetry reductions. Note that $\dim (\End^{G}(\cH) ) = \sum_{i=1}^{t}m_{i}^{2}$ can be significantly smaller than $\dim \cH$. Moreover, by this lemma, for any $A\in \End^{G}(\cH)$, the task of checking whether $A$ is a positive semidefinite matrix can be reduced to checking if the smaller $m_i\times m_i$ matrices $\llbracket \phi(A) \rrbracket_i$ are positive semidefinite, for every $i\in[t]$. The mapping $\phi$ in Eq.~\eqref{eq:block_diagonal_matrix} is a special case of the $*$-isomorphism of Theorem~\ref{thm:*-isomorphism decomposition}, where $\cA$ is the matrix $*$-algebra $\End^{G}(\cH)$.

\subsection{Representation theory of the symmetric group}

Fix $k \in \NN$ and a finite-dimensional vector space $\cH$ with $\dim(\cH) = d$. We consider the natural action of the symmetric group $\mathfrak{S}_k$ on $\cH^{\otimes k}$ by permuting the indices, i.e.,
\begin{align*}
	\pi \cdot (h_1 \otimes \dots \otimes h_k) = h_{\pi^{-1}(1)}\otimes \dots \otimes h_{\pi^{-1}(k)} \, , h_i \in \cH \, , \forall \pi \in \mathfrak{S}_k \, .
\end{align*}

Based on classical representation theory of the symmetric group, we describe a representative set for the action of $\mathfrak{S}_k$ on $\cH^{\otimes k}$. The concepts and notation we introduce in this section will be used throughout this paper.

A \emph{partiton} $\lambda$ of $k$ is a sequence $(\lambda_1,\dots,\lambda_d)$ of natural numbers with $\lambda_1 \geq \dots \lambda_d>0$ and $\lambda_1+\dots+\lambda_d = k$. The number $d$ is called the \emph{height} of $\lambda$. We write $\lambda \vdash_{d}k$ if $\lambda$ is a partition of $k$ with height $d$. Let $\mathrm{Par}(d,k) \coloneqq \{ \lambda\,:\,\lambda \vdash_{d} k \}$. The \emph{Young shape} $Y(\lambda)$ of $\lambda$ is the set
\begin{align*}
	Y(\lambda)\coloneqq \{(i,j) \in \NN^2: 1 \leq j \leq d, 1 \leq i \leq \lambda_j \} \, .
\end{align*}

Following the French notation~\cite{procesi2007lie}, for an index $j_0 \in [d]$, the $j_0$-th \emph{row} of $Y(\lambda)$ is set of elements $(i,j_0)$ in $Y(\lambda)$. Similarly, fixing an element $i_0 \in [\lambda_1]$, the $i_0$-th \emph{column} of $Y(\lambda)$ is set of elements $(i_0,j)$ in $Y(\lambda)$. We label the elements in $Y(\lambda)$ from $1$ to $k$ according the lexicographic order on their positions. Then the \emph{row stabilizer} $R_{\lambda}$ of $\lambda$ is the group of permutations $\pi$ of $Y(\lambda)$ with $\pi(L) = L$ for each row $L$ of $Y(\lambda)$. Similarly, the \emph{column stabilizer} $C_{\lambda}$ of $\lambda$ is the group of permutations $\pi$ of $Y(\lambda)$ with $\pi(L) = L$ for each column $L$ of $Y(\lambda)$. 
\par 
For $\lambda \vdash_{d} k$, a \emph{$\lambda$-tableau} is a function $\tau: Y(\lambda) \to \NN$. A $\lambda$-tableau is \emph{semistandard} if the entries are non-decreasing in each row and strictly increasing in each column. Let $T_{\lambda,d}$ be the collection of semistandard $\lambda$-tableaux with entries in $[d]$. We write $\tau \sim \tau'$ for $\lambda$-tableaux $\tau,\tau'$ if $\tau' = \tau r$ for some $r \in R_{\lambda}$. Let $e_1,\dots,e_d$ be the standard basis of $\cH$. For any $\tau \in T_{\lambda,d}$, define $u_{\tau}\in \cH^{\otimes k}$ as
\begin{align}
	u_{\tau} \coloneqq \sum_{\tau' \sim \tau}\sum_{c \in C_{\lambda}} \mathrm{sgn}(c) \bigotimes_{y \in Y(\lambda)}e_{\tau'(c(y))} \, .
\end{align}
Here the Young shape $Y(\lambda)$ is ordered by concatenating its rows. Then the matrix set
\begin{align}
	\label{eq:symmetric_representative_set}
	\{U_{\lambda}: \lambda \vdash_{d} k \} \, \, \text{ with } U_{\lambda} = [u_{\tau}: \tau \in T_{\lambda,d}]
\end{align}

is a representative for the natural action of $\mathfrak{S}_k$ on $\cH^{\otimes k}$ \cite[Section 2.1]{litjens2017semidefinite}. Moreover, we have
\begin{align}
	\label{eq:number_partitions}
	|\mathrm{Par}(d,k)| \leq (k+1)^d \text{ and } |T_{\lambda,d}| \leq (k+1)^{d(d-1)/2} \, \, , \forall \lambda \in \mathrm{Par}(d,k) \enspace.
\end{align}
   

\section{Efficient approximation of the regularized divergence of channels}
\label{section:4}
For $\alpha \in (1,\infty)$, the regularized sandwiched $\alpha$-R\'enyi divergence between channels $\cN_{X \to Y}$ and $\cM_{X \to Y}$ is defined as
\begin{align}
\label{eq:def_reg_div}
\widetilde{\rD}_{\alpha}^{\reg}(\cN \| \cM) := \lim_{k \to \infty} \frac{1}{k} \widetilde{\rD}_{\alpha}(\cN^{\otimes k} \| \cM^{\otimes k}) \ .
\end{align}

The regularized sandwiched R\'enyi divergence between channels can be used to obtain improved characterization of many information processing tasks such as channel discrimination~\cite{Fang_19,Omar_20}. However, the sandwiched R\'enyi divergence between channels is non-additive in general \cite{fang2020chain} and it is unclear whether its regularization can be computed efficiently. Ref.~\cite{Omar_20} provides a converging hierarchy of upper bounds on the regularized divergence between channels:
\begin{theorem}[\cite{Omar_20}]
	\label{thm:comp_reg_sand}
	Let $\alpha \in (1, \infty)$ and $\cN, \cM$ be completely positive maps from $\Lin(X)$ to $\Lin(Y)$. Then for any $k \geq 1$,
	\begin{align*}
	\frac{1}{k} \newD_{\alpha}(\cN^{\otimes k} \| \cM^{\otimes k}) - \frac{1}{k} \frac{\alpha}{\alpha-1} (d^2+d) \log (k+d) \leq \: \widetilde{\rD}_{\alpha}^{\reg}(\cN \| \cM) \: \leq \frac{1}{k} \newD_{\alpha}(\cN^{\otimes k} \| \cM^{\otimes k}) \ ,
	\end{align*}
	where $d = \dim X \dim Y$.
\end{theorem}

We note that $\frac{1}{k} \newD_{\alpha}(\cN^{\otimes k} \| \cM^{\otimes k})$ is decreasing in $k$ (since the $\newD$ channel divergence is subadditive). Moreover, $ \newD_{\alpha}(\cN^{\otimes k} \| \cM^{\otimes k})$ can be written in terms of a  convex program as (\cite{Omar_20})
\begin{align}
\label{sdp:regD_approximation}
 \frac{1}{\alpha-1} \log \min_{A_{X^{\otimes k}Y^{\otimes k}} \geq 0}   \| \tr_{Y^{\otimes k}}(A_{X^{\otimes k}Y^{\otimes k}}) \|_{\infty}  \quad \textup{ s.t. } \quad  (J^{\cN^{\otimes k}}) \leq (J^{\cM^{\otimes k}}) \#_{1/\alpha} A_{X^{\otimes k}Y^{\otimes k}} .
\end{align}
Therefore, Theorem~\ref{thm:comp_reg_sand} establishes that $\widetilde{\rD}_{\alpha}^{\reg}(\cN \| \cM)$ can be approximated by $\frac{1}{k}\newD_{\alpha}(\cN^{\otimes k}\|\cM^{\otimes k})$ with arbitrary accuracy for sufficiently large $k$ in finite time. Namely, if we take $k = \ceil{\frac{8\alpha d^3}{(\alpha-1)\epsilon}}$ then we have
\begin{align*}
	|\widetilde{\rD}_{\alpha}^{\reg}(\cN \| \cM)-\frac{1}{k} \newD_{\alpha}(\cN^{\otimes k} \| \cM^{\otimes k})| \leq \epsilon \, .
\end{align*} 
However, the size of Program~\eqref{sdp:regD_approximation} grows exponentially with $k$. 


\subsection{Exploiting symmetries to simplify the problem}

In this section, we will show how the symmetries of Program~\eqref{sdp:regD_approximation} can be used to simplify this optimization problem and solve it in time polynomial in $k$. We first focus on the natural symmetries arising due to invariance under permutation of physical systems. In Section~\ref{sec:Zsymmetry}, we show how additional symmetries can be utilized to further simplify the problem. Our approach can be summarized as follows: First, we show that program~\eqref{sdp:regD_approximation} is invariant with respect to the action of the symmetric group. Using this observation, we show that the program can be transformed into an equivalent program with polynomially many constraints, each of polynomial size in $k$. In order to show this, we use the block-diagonal decomposition given by Lemma~\ref{lemma:psd_preserving}. A naive implementation of this transformation, however, involves exponential time computations. We show that the simplified form of the program can be directly computed in $\poly(k)$ time. 

Recall that, for every $\pi \in \mathfrak{S}_k$, we consider the action of $\pi$ on $k$ copies of a finite dimensional Hilbert space $\cH$ as
\begin{align*}
\pi \cdot \left( h_1 \otimes \dots \otimes h_k \right) =  h_{\pi^{-1}(1)} \otimes \dots \otimes h_{\pi^{-1}(k)} \quad,\quad h_i \in \cH \,,\, \forall i\in [k] \enspace.
\end{align*}
 Let $P_X(\pi)$ and $P_Y(\pi)$ be the permutation matrices corresponding to the action of $\pi$ on $X^{\otimes k}$ and $Y^{\otimes k}$, respectively. Note that the action of $\pi$ on $(X\otimes Y)^{\otimes k}$ corresponds to the simultaneous permutation of the $X$ and $Y$ tensor factors and the corresponding permutation matrix, when the subsystems are reordered as $X^{\otimes k} \otimes Y^{\otimes k}$, is given by $P_{X\otimes Y}(\pi) = P_X(\pi) \otimes P_Y(\pi)$.

 The following lemma shows that the feasible region of the convex program \eqref{sdp:regD_approximation} may be restricted to the permutation invariant algebra of operators on $X^{\otimes{k}} \otimes Y^{\otimes{k}}$ without changing the optimal value.

 For a linear operator $X\in\Lin(\cH^{\otimes k})$, we define its \emph{group average} operator denoted $\overline{X}$ as 
 \begin{align*}
 	\overline{X} \coloneqq \frac{1}{|\mathfrak{S}_k|}\sum_{\pi \in \mathfrak{S}_k} P_\cH(\pi) X P_\cH(\pi)^{*} \enspace.
 \end{align*}
 
 \begin{lemma} \label{lem:inv_space_1}
 	The convex program of Eq.~\eqref{sdp:regD_approximation} has an optimal solution $A \in \End^{\mathfrak{S}_k} \br{ X^{\otimes k}\otimes Y^{\otimes k} }$.
 \end{lemma}
 
 \begin{proof}
 	It is straightforward to check that by Slater's condition the optimal value is achieved by a feasible solution. We will prove that for every feasible solution $A$, the corresponding group-average operator $\overline{A}$ is a feasible solution with an objective value not greater than the original value.
 	
 	To simplify the notation, let $\Pi(\pi)\coloneqq P_{X\otimes Y}(\pi)$. We have
 	\begin{align}
 		\overline{A} \; \#_{1/\alpha} \; J^{\cM^{\otimes k}} &= 
 		\br{ \frac{1}{|\mathfrak{S}_k|} \sum_{\pi \in \mathfrak{S}_k} \Pi(\pi) A \Pi(\pi)^{*} } \#_{1/\alpha} \br{ \frac{1}{|\mathfrak{S}_k|} \sum_{\pi \in \mathfrak{S}_k} \Pi(\pi) J^{\cM^{\otimes k}} \Pi(\pi)^{*} } \label{eq:line_1}\\
 		&\geq \sum_{\pi \in \mathfrak{S}_k} \br{ \frac{1}{|\mathfrak{S}_k|} \Pi(\pi) A \Pi(\pi)^{*} } \#_{1/\alpha} \br{ \frac{1}{|\mathfrak{S}_k|} \Pi(\pi) J^{\cM^{\otimes k}} \Pi(\pi)^{*} } \label{eq:line_2}\\
 		&= \frac{1}{|\mathfrak{S}_k|} \sum_{\pi \in \mathfrak{S}_k} \Pi(\pi) \br{A\#_{1/\alpha} J^{\cM^{\otimes k}} } \Pi(\pi)^{*} \label{eq:line_3}\\
 		&\geq  \frac{1}{|\mathfrak{S}_k|} \sum_{\pi \in \mathfrak{S}_k} \Pi(\pi) \br{ J^{\cN^{\otimes k}} } \Pi(\pi)^{*} \label{eq:line_4}\\
 		& = J^{\cN^{\otimes k}} \label{eq:line_5}\enspace,     
 	\end{align}
 	where Eq.~\eqref{eq:line_1} holds since $J^{\cM^{\otimes k}} \in \End^{\mathfrak{S}_k} \br{ X^{\otimes k}\otimes Y^{\otimes k} }$, inequality~\eqref{eq:line_2} follows from the joint-concavity property of the geometric mean, Eq.~\eqref{eq:line_3} is a consequence of properties $2$ and $3$ of the geometric mean, inequality~\eqref{eq:line_4} holds by feasibility of $A$, and finally, Eq.~\eqref{eq:line_5} follows since  $J^{\cN^{\otimes k}} \in \End^{\mathfrak{S}_k} \br{ X^{\otimes k}\otimes Y^{\otimes k} }$.
 	
 	For the objective function, note that since $\Pi(\pi)=P_X(\pi)\otimes P_Y(\pi)$, we have 
 	\begin{align*}
 		\tr_{Y^{\otimes k}} \left(\Pi(\pi)A\Pi(\pi)^{T} \right) = P_X(\pi) \,\tr_{Y^{\otimes k}}(A)\, P_X(\pi)^{T} \enspace.
 	\end{align*}
 	Therefore, by the triangle inequality and the unitary invariance of the operator norm, we have
 	\begin{align*}
 		\| \tr_{Y^{\otimes k}} \left( \overline{A} \right) \|_{\infty} &=
 		\left\| \tr_{Y^{\otimes k}} \left(\frac{1}{|\mathfrak{S}_k|} \sum_{\pi \in \mathfrak{S}_k} \Pi(\pi)A\Pi(\pi)^{T} \right)\right\|_{\infty} \\  
 		&= \left\| \frac{1}{|\mathfrak{S}_k|}\sum_{\pi \in \mathfrak{S}_k}P(\pi_X)\tr_{Y^{\otimes k}}(A)P(\pi_X)^{T} \right\|_{\infty}\\
 		&\leq \frac{1}{|\mathfrak{S}_k|}\sum_{\pi \in \mathfrak{S}_k} \left\|P(\pi_X) \tr_{Y^{\otimes k}}(A)P(\pi_X)^{T} \right\|_{\infty}\\
 		&= \|\tr_{Y^{\otimes k}}(A)\|_{\infty}.
 	\end{align*}
 	This concludes the proof.
 \end{proof}
 
Recall that in the convex program~\eqref{sdp:regD_approximation}, the number of the variables and the size of the PSD constraints grow exponentially with $k$. Using the observation made in Lemma~\ref{lem:inv_space_1}, we show that this optimization problem can be transformed into a form having a number of variables and constraints that is polynomial in $k$. Before doing so, we introduce some notation.

Let $\cH \in \{X,Y,X\otimes Y\}$ and $d_\cH \coloneqq \dim \cH$. The algebra of $\mathfrak{S}_k$-invariant operators on $\cH^{\otimes k}$ is given by
\begin{align*}
	\End^{\mathfrak{S}_k} \left( \cH^{\otimes k} \right) &= \{A \in \Lin(\cH^{\otimes k}): P_\cH(\pi)\, A \, P_\cH(\pi)^* = A, \; \forall \pi \in \mathfrak{S}_k \} \, .
\end{align*}
Let $\phi_{\cH}$ denote the linear map defined in Eq.~\eqref{eq:block_diagonal_matrix} that maps the elements of $\End^{\mathfrak{S}_k} \left( \cH^{\otimes k} \right)$ into block-diagonal form:
\begin{align}
	\phi_{\cH}: \End^{\mathfrak{S}_k} \left( \cH^{\otimes k} \right) &\to \bigoplus_{\lambda \in \mathrm{Par}(d_{\cH},k)} \CC^{m_\lambda^\cH \times m_\lambda^\cH} \nonumber \\ 
	A &\mapsto \bigoplus_{\lambda\in \mathrm{Par}(d_{\cH},k)} \left(\<Au_{\gamma},u_{\tau}\>  \right)_{\tau,\gamma \in T_{\lambda,d_{\cH}}} \, .	\label{eq:phi_H}
\end{align} 
In this decomposition, the number of blocks and the size of the blocks are bounded by a polynomial in $k$. In particular, we have
\begin{align}
	t^\cH&\coloneqq |\mathrm{Par}(d_{\cH},k)| \leq (k+1)^{d_\cH} \enspace, \label{eq:partitions}\\
	m_\lambda^\cH  &\coloneqq |T_{\lambda,d_{\cH}}| \leq (k+1)^{d_\cH(d_\cH-1)/2} \;,\quad \forall \lambda \in \mathrm{Par}(d_{\cH},k) \; . \label{eq:dim_S_lambda}
\end{align} 
From Eqs. \eqref{eq:partitions} and \eqref{eq:dim_S_lambda}, we get 
\begin{equation}
	\label{eq:dim_Sk_inv}
    m^\cH\coloneqq \dim \left[ \End^{\mathfrak{S}_k}(\cH^{\otimes k}) \right] \leq (k+1)^{d_\cH^2}
\end{equation}

\begin{theorem}
  	\label{thm: computing_channel_divergence_1}
  	The channel divergence $\newD_{\alpha}(\cN^{\otimes k} \| \cM^{\otimes k})$ can be formulated as a convex program with $\mathrm{O}\!\br{k^{d^2}}$ variables and $\mathrm{O}\!\br{k^d}$ PSD constraints involving matrices of size at most $(k+1)^{d(d-1)/2}$, where $d=d_X d_Y$. 
\end{theorem} 

\begin{proof}
  	By Lemma~\ref{lem:inv_space_1} and Property $2$ of the $\alpha$-geometric mean, after a permutation of the $X$ and $Y$ tensor factors, the formulation \eqref{sdp:regD_approximation} for $\newD_{\alpha}(\cN^{\otimes k} \| \cM^{\otimes k})$ can be written as
  	\begin{IEEEeqnarray}{rCl}
  	\frac{1}{\alpha - 1}\,\log \quad \min_{A,y}& \quad  &y\\
  	\mathrm{ s.t. }& &\tr_{Y^{\otimes k}}(A) \leq y\,\id_{X^{\otimes k}} \enspace, \label{eq:constraint_obj}\\
  	& &\br{J^{\cN}}^{\otimes k} \leq \br{J^{\cM}}^{\otimes k} \#_{1/\alpha} A \enspace, \label{eq:constraint_geometric_mean}
  	\end{IEEEeqnarray}
  	where $A \in \Pos\br{\br{X\otimes Y}^{\otimes k}} \cap \End^{\mathfrak{S}_k} \br{\br{X\otimes Y}^{\otimes k}}$ and $y\in \mathbb{R}$. 
  	
  	For $\cH\in\{X,X\otimes Y\}$, let $\phi_\cH: \End^{\mathfrak{S}_k}(\cH^{\otimes k}) \to \bigoplus_{i=1}^{t^\cH} \CC^{m_i^\cH \times m_i^\cH}$ be the bijective linear map defined in Eq.~\eqref{eq:phi_H} which block-diagonalizes the corresponding invariant algebra, where to simplify the notation, the blocks are indexed by $i\in[t^\cH]$ instead of $\lambda\in \mathrm{Par}(d_\cH,k)$. For $Z\in \End^{\mathfrak{S}_k}(\cH^{\otimes k})$, we denote the $i$-th block of $\phi_\cH(Z)$ by $\llbracket \phi_\cH(Z) \rrbracket_i$.  Note that by Lemma~\ref{lemma:psd_preserving}, $\phi_\cH$ preserves positive semidefiniteness. Therefore, since $\tr_{Y^{\otimes k}}(A)$, $\id_{X^{\otimes k}} \in \End^{\mathfrak{S}_k}(X^{\otimes k})$, the constraint~\eqref{eq:constraint_obj} can be mapped by $\phi_X$ into the direct sum form. By Lemma~\ref{lemma:psd_preserving} and Property~$2$ of the $\alpha$-geometric mean, we have $\phi_{X\otimes Y}\br{ \br{J^{\cM}}^{\otimes k} \#_{1/\alpha} A }=\phi_{X\otimes Y}\br{ \br{J^{\cM}}^{\otimes k} } \#_{1/\alpha} \phi_{X\otimes Y}(A)$. Therefore, by Property~$5$ of the $\alpha$-geometric mean (direct sum property), the constraint~\eqref{eq:constraint_geometric_mean} can be decomposed into constraints involving the smaller diagonal blocks as well. The transformed convex program can be written as
  	\begin{IEEEeqnarray*}{rCl"r}
  	\frac{1}{\alpha - 1}\,\log \quad \min& \quad &y &\\
  	\mathrm{ s.t.}&  &\left\llbracket \br{\phi_X \circ \tr_{Y^{\otimes k}} \circ \phi_{X\otimes Y}^{-1}} \br{ \oplus_l A_l} \right\rrbracket_j \leq y\,\id_{m_j^X} \enspace, &j\in \Br{t^X}\\
  	& &\left\llbracket \phi_{X\otimes Y} \br{\br{J^{\cN}}^{\otimes k}} \right\rrbracket_i \leq \left\llbracket \phi_{X\otimes Y}\br{\br{J^{\cM}}^{\otimes k}} \right\rrbracket_i \#_{1/\alpha} A_i \enspace, &i \in \Br{t^{X\otimes Y}}\\
  	& &A_i \in \Pos\br{\CC^{m_i^{X\otimes Y}}}, \qquad &i \in \Br{t^{X\otimes Y}}
  	\end{IEEEeqnarray*}
  	
  	 The statement of the theorem follows since for $\cH\in\{X,X\otimes Y\}$, by Eq.~\eqref{eq:partitions}, we have $t^\cH \leq (k+1)^{d_\cH}$ and by Eq.~\eqref{eq:dim_S_lambda}, for every $i\in \Br{t^\cH}$, we have $m_i^\cH \leq (k+1)^{d_\cH(d_\cH-1)/2}$.
\end{proof}

Note that a direct implementation of the transformation mapping the convex program~\eqref{sdp:regD_approximation} into the polynomial-size form of Theorem~\ref{thm: computing_channel_divergence_1} involves exponential computations. Next, we show how to do this efficiently.  

\bigskip
\textbf{A basis for the invariant subspace}. The canonical basis of the matrix $*$-algebra $\End^{\mathfrak{S}_k}\left(\cH^{\otimes k} \right)$ consists of zero-one incidence matrices of orbits of the group action on pairs (see~\cite{Klerk_07, Bachoc2012} for more information). In particular, let the standard basis of $\cH^{\otimes k}$ be indexed by $i \in \left[(d_\cH)^k\right]$. Then the orbit of the pair $(i,j) \in \left[(d_\cH)^k\right]^2$ under the action of the group $\mathfrak{S}_k$ is given by
\begin{align*}
	O(i,j) = \{(\pi(i),\pi(j)): \pi \in \mathfrak{S}_k \},
\end{align*}   
where $\pi(i)$ is the index of the basis vector $P_\cH(\pi)\ket{i}$. With this notation, for every $A \in \End^{\mathfrak{S}_k}\left(\cH^{\otimes k} \right)$, and every $\pi \in \mathfrak{S}_k$, we have $A_{ij}  = A_{\pi^{-1}(ij)}= A_{\pi^{-1}(i),\pi^{-1}(j)}$. The set $\left[(d_\cH)^k\right]^2$ decompose into orbits $O_{1}^{\cH},\dots,O_{m^\cH}^{\cH}$ under the action of $\mathfrak{S}_k$.
 For each $r \in [m^\cH]$, we construct a zero-one matrix $C_{r}^{\cH}$ of size $(d_\cH)^k \times (d_\cH)^k$ given by
\begin{align}
	\label{eq:canonical_basis}
	(C_{r}^{\cH})_{ij} = 
	\begin{cases} 
		1 & \text{if }  (i,j) \in O_{r}^{\cH} \, , \\
		0 & \text{otherwise}.
	\end{cases}
\end{align}
The set $\cC^\cH=\{C_{1}^{\cH},\dots,C_{m^\cH}^{\cH}\}$ forms an orthogonal basis of $\End^{\mathfrak{S}_k}\left(\cH^{\otimes k} \right)$ with $m^\cH \leq (k+1)^{d_\cH^2}$.

\medskip 
\textbf{Enumerating all orbits}. For each $r= 1,\dots,m^{\cH}$, we need to compute a representative element of $O_{r}^{\cH}$. In order to do so, we define a matrix $E^{(i,j)} \in \mathbb{Z}_{\geq 0}^{d_\cH \times d_\cH}$
\begin{align}
	\label{eq:number_pair}
	(E^{(i,j)})_{a,b} \coloneqq \left|\{v \in [k]: i_v = a,j_v = b \}\right| \; , \; \forall a,b \in[d_{\cH}] \, .
\end{align}  
  
By the construction in Eq.~\eqref{eq:number_pair}, for two pairs $(i,j),(i',j') \in[d_\cH]^k \times [d_\cH]^k$, we have $(i',j') = (\pi(i),\pi(j))$, for some $\pi \in \mathfrak{S}_k$ if and only if $E^{(i,j)} = E^{(i',j')}$. Therefore, there is a one-to-one correspondence between the orbits $\cbr{O_{r}^{\cH}}_{r\in[m^{\cH}]}$ and $E \in \mathbb{Z}_{\geq 0}^{d_\cH \times d_\cH}$ such that $\sum_{a,b}E_{a,b} = k$. Therefore, we can determine a representative element for every $O_{r}^{\cH}$ in $\poly(k)$ time by listing  all non-negative integer solutions of the equation $\sum_{a,b \in [d_\cH]}E_{a,b} = k$. 

Any matrix in $\End^{\mathfrak{S}_k}(\cH^{\otimes k})$ can be written in the basis $\cC^\cH$ as
\begin{align}
	M(z) \coloneqq \sum_{r=1}^{m^{\cH}}z_{r} C_{r}^{\cH} \, , \text{ for some } z \in \CC^{[m^{\cH}]} \, .
\end{align}

Using the representative matrix for the action of $\mathfrak{S}_k$ on the space $\cH^{\otimes k}$ in Eq.~\eqref{eq:symmetric_representative_set}, we get
\begin{align}
	\label{eq:block_Phi}
	\phi_{\cH}(M(z)) =  \sum_{r=1}^{m^{\cH}}z_{r} \phi_{\cH}\br{C_{r}^{\cH}} =  \sum_{r=1}^{m^{\cH}} z_{r} \bigoplus_{\lambda \vdash_{d_{\cH}}k} U_{\lambda}^{T}C_{r}^{\cH}U_{\lambda} \, .
\end{align}
Note that $U_{\lambda}$ is real matrix for all $\lambda \in \mathrm{Par}(d_{\cH},k)$.

We show that, for every $r\in[m^\cH]$, $\phi_{\cH}(C_{r}^{\cH})$ can be computed in $\poly(k)$ time. In order to do so, we show how to efficiently compute each block $U_{\lambda}^{T}C_{r}^{\cH}U_{\lambda}$ indexed by $\lambda \in \mathrm{Par}(d_\cH,k)$. This in fact boils down to efficiently computing $u_{\tau}^{T}C_{r}^{\cH} u_{\gamma}$, for every $\tau,\gamma \in T_{\lambda, d_{\cH}}$. We note that $u_{\tau}$, $u_{\gamma}$, and $C_{r}^{\cH}$ all have exponential size in $k$.
\par
For $\cH \in \{X,Y,X\otimes Y \}$, let $W_{\cH} \coloneqq \cH \otimes \cH$. For every $p=(i,j)\in [d_\cH]^2$, define
\begin{align*}
	a_{p} \coloneqq e_i \otimes e_j \in W_{\cH} \, , 
\end{align*}
where $\cbr{e_i}_{i\in[d_\cH]}$ is the standard basis of $\cH$. Then the set $\mathcal{B}\coloneqq \cbr{a_{p}: p \in [d_{\cH}]^2 }$ is a basis of $W_{\cH}$. Let $\mathcal{B^*}\coloneqq \cbr{a^*_{p}: p \in [d_{\cH}]^2 }$ be the corresponding dual basis for $W^*_\cH$. 

Using the natural identification of $\left[(d_{\cH})^k \right]^2$ and $(\left[d_{\cH}\right]^2)^k$, for every $r \in [m^{\cH}]$, we map $O^\cH_r\subseteq \left[(d_{\cH})^k \right]^2$ to $\mathsf{O}^\cH_r \subseteq ([d_{\cH}]^2)^k$. Then corresponding to each operator $C_{r}^{\cH}$, we define
\begin{align}
	\mathsf{C}_{r}^{\cH} \coloneqq \sum_{(p_1,\dots,p_k) \in \mathsf{O}^\cH_{r}} a_{p_1} \otimes \dots \otimes a_{p_k} \in W_{\cH}^{\otimes k} \, .
\end{align}
Note that $\mathsf{C}_{r}^{\cH}$ can be obtained from $\mathrm{vec}\br{C_{r}^{\cH}}$ by applying the permutation operator which maps $\br{\cH^{\otimes k}}^{\otimes 2}$ to $\br{\cH^{\otimes 2}}^{\otimes k}$.
For every $(p_1,\dots,p_k)\in \left[(d_{\cH})^2 \right]^k$, let 
\begin{equation}
    m(p_1,\dots,p_k)\coloneqq a_{p_1}^{*}\cdots a_{p_k}^{*} \in \mathcal{O}_{k}(W_{\cH})\enspace
\end{equation} 
be a degree $k$ monomial expressed in the basis $\mathcal{B}^*$. Note that, for a fixed $r\in[m^\cH]$, $m(p_1,\dots,p_k)$ is the same monomial, for every $(p_1,\dots,p_k)\in \mathsf{O}^\cH_r$. We denote this monomial by $m\br{\mathsf{O}^\cH_r}$. Moreover, $\cbr{\mathsf{O}^\cH_r}_{r\in [m^\cH]}$ partitions $\left[(d_{\cH})^2 \right]^k$ into disjoint subsets. Therefore, there is a bijection between $\cbr{\mathsf{O}_{i}^{\cH}}_{i\in [m^\cH]}$ and the set of degree $k$ monomials expressed in the basis $\mathcal{B}^*$.

Let $\zeta:(W_{\cH}^{*})^{\otimes k} \to \mathcal{O}_{k}(W_{\cH})$ be the linear map defined as
\begin{align*}
	\zeta(w_1^{*}\otimes \dots \otimes w_{k}^{*}) \coloneqq w_1^{*} \cdots w_{k}^{*}\;,\; \forall w_1^{*},\dots,w_{k}^{*} \in W_\cH^{*} \, .
\end{align*}
To simplify the notation we write $\overline{w} = \zeta(w)$, for every $w \in (W_{\cH}^{*})^{\otimes k}$. 

For every $\lambda\in \mathrm{Par}(d_\cH,k)$ and $\tau,\gamma \in T_{\lambda,d_{\cH}}$, define the polynomial $f_{\tau,\gamma} \in \CC[x_{i,j}: i,j\in [d_{\cH}]]$ by 
\begin{align}
	\label{eq:polynomial}
	f_{\tau,\gamma}(X) \coloneqq \sum_{\substack{\tau' \sim \tau \\ \gamma' \sim \gamma }} \sum_{c,c' \in C_{\lambda}} \mathrm{sgn}(cc') \prod_{y \in Y(\lambda)} x_{\tau'c(y),\gamma'c'(y)} \, ,
\end{align} 
for $X = (x_{i,j})_{i,j = 1}^{d_{\cH}} \in \CC^{d_{\cH} \times d_{\cH}}$. Refs.~\cite[Proposition 3]{litjens2017semidefinite} and \cite[Theorem 7]{gijswijt2009block} show that the polynomial in Eq.~\eqref{eq:polynomial} can be computed (i.e., expressed as a linear combination of monomials in variables $x_{i,j}$) in polynomial time. 
\begin{lemma}
	\label{lem:polynomial}
	For every $\lambda \in \mathrm{Par}(d_\cH,k)$ and every $\tau,\gamma \in T_{\lambda,d_{\cH}}$, expressing the polynomial $f_{\tau,\gamma}(X)$ as a linear combination of monomials can be done in $\poly(k)$ time, for fixed $d_{\cH}$. 
\end{lemma}   

We use this to prove the following lemma:

\begin{lemma}[Lemma 2, \cite{litjens2017semidefinite}]
	\label{lem:computing_entries_block}
	Let $\lambda \in \mathrm{Par}(d_\cH,k)$, $\tau,\gamma \in T_{\lambda,d_{\cH}}$, and $r\in [m^\cH]$. Then $u_{\tau}^{T}C_{r}^{\cH}u_{\gamma}$ can be computed in polynomial time in $k$, for fixed $d_{\cH}$.   
\end{lemma}
\begin{proof}
The proof can be found in~\cite{litjens2017semidefinite}, but we include a concise proof for the reader's convenience. 
For every $r \in [m^{\cH}]$, it is straightforward to see that $u_{\tau}^{T}C_{r}^{\cH}u_{\gamma} = (u_{\tau} \otimes u_{\gamma})^T\mathrm{vec}\br{C_{r}^{\cH}}$. Therefore, by a permutation of the tensor factors, we get $u_{\tau}^{T}C_{r}^{\cH}u_{\gamma} = w\, \mathsf{C}_{r}^{\cH}$, for $w\in \br{W^*_\cH}^{\otimes k}$ given by

\begin{align*}
	w  =  \sum_{\substack{\tau' \sim \tau \\ \gamma' \sim \gamma}}\sum_{c,c' \in C_{\lambda}} \mathrm{sgn}(cc') \bigotimes_{y \in Y(\lambda)} (A)_{\tau'\br{c(y)},\gamma'\br{c'(y)}} \, , 
\end{align*} 
where $A \in (W^{*})^{d_{\cH} \times d_{\cH}}$ with $(A)_{x,y} = a_{(x,y)}^{*}$. Then
\begin{align*}
	\sum_{r\in [m^{\cH}]} \br{u_{\tau}^{T}C_{r}^{\cH}u_{\gamma} } m(\mathsf{O}_{r}^{\cH}) &= \sum_{r\in [m^{\cH}]} \br{ w\, \mathsf{C}_{r}^{\cH} } m(\mathsf{O}_{r}^{\cH}) \\
	&= \sum_{(p_1,\dots,p_k) \in ([d_{\cH}]^2)^k} \br{w\, (a_{p_1} \otimes \dots \otimes a_{p_k})}\, a_{p_1}^{*}\cdots a_{p_{k}}^{*} \\
	&= \overline{w} = \sum_{\substack{\tau' \sim \tau \\ \gamma' \sim \gamma}} \sum_{c,c' \in C_{\lambda}} \mathrm{sgn}(cc') \prod_{y \in Y(\lambda)} (A)_{\tau'\br{c(y)},\gamma'\br{c'(y)}}  \\
	&= f_{\tau,\gamma}(A) \, .
\end{align*}

Therefore, $u_{\tau}^{T}C_{r}^{\cH}u_{\gamma}$ is exactly the coefficient of the monomial $m(\mathsf{O}_{r}^{\cH})$ in $f_{\tau,\gamma}(A)$, which by Lemma~\ref{lem:polynomial} can be computed in $\poly(k)$ time.
\end{proof}

\begin{theorem}
		\label{thm: computing_channel_divergence}
	There exists an algorithm which given as input $J^\cM$, $J^\cN$, and $k\in \NN$, outputs in $\poly(k)$ time (for fixed $\dim(X \otimes Y)$) the description of a convex program of size described in Theorem~\ref{thm: computing_channel_divergence_1} for computing $\newD_{\alpha}(\cN^{\otimes k} \| \cM^{\otimes k})$.  
\end{theorem}
  
\begin{proof}
    For $\cH \in \{X\otimes Y, X\}$, let $\cbr{O_{r}^{\cH}}_{r\in[m^\cH]}$ denote the set of orbits of pairs and $\cbr{C_{r}^{\cH}}_{r\in[m^\cH]}$ denote the canonical basis of $\End^{\mathfrak{S}_k}\left(\cH^{\otimes k} \right)$ defined in Eq.~\eqref{eq:canonical_basis}. For every $r\in [m^{X \otimes Y}]$, we define $D_r\coloneqq \tr_{Y^{\otimes k}}\br{C_r^{X \otimes Y}}$. Note that $D_r\in \End^{\mathfrak{S}_k} \br{X^{\otimes k}}$. Then by Theorem~\ref{thm: computing_channel_divergence_1}, $\newD_{\alpha}(\cN^{\otimes k} \| \cM^{\otimes k})$ can be formulated as the following convex program:
    
  	\begin{IEEEeqnarray*}{rCl"r}
  	\frac{1}{\alpha - 1}&\, &\log \; \min \quad y &\\
  	\mathrm{s.t.}&  &\sum_{r=1}^{m^{X \otimes Y}} z_r \left\llbracket \phi_X(D_r) \right\rrbracket_j \leq y\,\id_{m_j^X} \enspace, &j\in \Br{t^X}\\
  	& &\left\llbracket \phi_{X\otimes Y} \br{\br{J^{\cN}}^{\otimes k}} \right\rrbracket_i \leq \left\llbracket \phi_{X\otimes Y}\br{\br{J^{\cM}}^{\otimes k}} \right\rrbracket_i \#_{1/\alpha} \sum_{r=1}^{m^{X \otimes Y}} z_r \left\llbracket \phi_{X \otimes Y}(C_r^{X \otimes Y}) \right\rrbracket_i \enspace, &i \in \Br{t^{X\otimes Y}}\\
  	& &\sum_{r=1}^{m^{X \otimes Y}} z_r \left\llbracket \phi_{X \otimes Y}(C_r^{X \otimes Y}) \right\rrbracket_i \geq 0\enspace, &i \in \Br{t^{X\otimes Y}} \\
  	& &y,z_r\in \RR \enspace, \qquad &r \in \Br{m^{X\otimes Y}}
  	\end{IEEEeqnarray*}
  	
    Here, we use the notation introduced in Theorem~\ref{thm: computing_channel_divergence_1}. 
    By Lemma~\ref{lem:computing_entries_block}, the block diagonal matrices $\phi_{X\otimes Y}(C^\cH_r)$ can be computed in $\poly(k)$ time, for every $r\in[m^\cH]$. Therefore, to complete the proof it suffices to show how to expand $\br{J^{\cN}}^{\otimes k}, \br{J^{\cM}}^{\otimes k}$ in the basis $\cbr{C_r^{X \otimes Y}}_{r\in [m^{X \otimes Y}]}$ and $D_r$ in the basis $\cbr{C_r^X}_{r\in [m^X]}$, for every $r\in [m^{X \otimes Y}]$. 
  	  
    For $\br{J^{\cM}}^{\otimes k} \in \End^{\mathfrak{S}_k} \left( (X \otimes Y)^{\otimes k}\right)$, if we take an arbitrary representative element $(p_1,\ldots, p_k) $ of $\mathsf{O}_{r}^{X \otimes Y}$, for every $r \in \Br{m^{X \otimes Y}}$, and define
  	  \begin{align*}
  	  	z_{r} \coloneqq \prod_{t=1}^{k}(J^{\cM})_{p_t} \, ,
  	  \end{align*} 
  	  then we have $\br{J^{\cM}}^{\otimes k} = \sum_{r=1}^{m^{X \otimes Y}}z_{r}C_{r}^{X \otimes Y}$. The same method can be used for $\br{J^{\cN}}^{\otimes k}$.
  	  
  	  Recall that, for every $r \in \Br{m^{X \otimes Y}}$, we have
  	  \begin{align*}
  	  	C_r^{X\otimes Y}   = \sum_{(i,j) \in O_r^{X \otimes Y}} \ketbra{i}{j} \, ,
  	  \end{align*}
  	    where $i= \br{i_{1}^{X}i_{1}^{Y}\cdots i_{k}^{X}i_{k}^{Y}}$ and $j= \br{j_{1}^{X}j_{1}^{Y}\cdots j_{k}^{X}j_{k}^{Y}}$. For any representative element $(i,j)$  of $O_{r}^{X \otimes Y}$ if $i^Y=(i_{1}^Y \cdots i_{k}^Y) \neq j^Y=j_{1}^Y \cdots j_{k}^Y$ then $\tr_{Y^{\otimes k}}\br{\ketbra{i}{j}}= 0$. Therefore,   
  	    \begin{equation*}
  	    	D_r = \sum_{\substack{(i,j) \in O_{r}^{X \otimes Y}\\ i^Y=j^Y}} \ketbra{i^X}{j^X}.
  	    \end{equation*}
       Moreover, for any representative element $(i,j)$ of $O_{r}^{X \otimes Y}$, we can determine the orbit $O_{t}^{X}$ that contains $(i^X,j^X)$ in $\poly(k)$ time. So if we define $\alpha \coloneqq |\{ \pi \in \mathfrak{S}_k: \pi (i^{X}) = i^{X}, \pi (j^X) = j^X \}|$, then
       \begin{align*}
       	D_{r} = \alpha \,C_{t}^{X} \, .
       \end{align*} 
      Furthermore, we have $\alpha = \prod_{a,b \in [d_X]}[(E^{(i^X,j^X)})_{a,b}]!$ with $E^{(i^X,j^X)} \in \mathbb{Z}_{\geq 0}^{d_X \times d_X}$ defined in Eq.~\eqref{eq:number_pair}. This concludes the proof.
\end{proof}


Alternatively, the regular $*$-representation approach can be used to show that the convex program~\eqref{sdp:regD_approximation} can be computed in $\poly(k)$ time. For $\cH\in\{X,X\otimes Y\}$, let $\psi_\cH$ be the regular $*$-representation of $\End^{\mathfrak{S}_k} \br{ \cH^{\otimes k}}$, defined explicitly in Theorem~\ref{thm: regular_rep}. We denote by $\{O_{r}^{\cH}\}_{r\in[m^\cH]}$ and $\{C_r^{\cH}\}_{r\in[m^\cH]}$, the orbits of pairs and the canonical basis of $\End^{\mathfrak{S}_k}(\cH^{\otimes k})$, following the construction in Eq.~\eqref{eq:canonical_basis}. The convex program can be reformulated as
\begin{IEEEeqnarray*}{rCl}
	\frac{1}{\alpha - 1}\,\log \quad \min& \quad  &y\\
	\mathrm{ s.t. }& &\textstyle\sum_{r=1}^{m^{X\otimes Y}} x_r \, \psi_X \! \br{{D_r}} \leq y \id_{m^X} \enspace,\\
	& &\psi_{X \otimes Y}\!\br{\br{J^{\cN}}^{\otimes k}} \leq \psi_{X \otimes Y}\!\br{\br{J^{\cM}}^{\otimes k}} \#_{1/\alpha} \textstyle\sum_{r=1}^{m^{X\otimes Y}} x_r \, \psi_{X \otimes Y}\!\br{C_r^{X \otimes Y}} \enspace,\\
	& &\textstyle\sum_r x_r \, \psi_{X \otimes Y}\!\br{C_r^{X \otimes Y}} \geq 0 \enspace,\\
	& & x_1,\ldots,x_{m^{X \otimes Y}},y\in \RR \enspace.
\end{IEEEeqnarray*}
Recall that $\psi_\cH\br{\End^{\mathfrak{S}_k} \br{ \cH^{\otimes k}}} \subseteq \CC^{m^\cH \times m^\cH}$, where $m^\cH \leq (k+1)^{d_\cH^2}$. 

Note that $\|C_r^{\cH}\|\coloneqq \sqrt{\langle C_r^{\cH},C_r^{\cH} \rangle}$ equals the size of the orbit $O_r^{\cH}$. Using the structure of the orbits, we can compute the multiplication parameters of $\End^{\mathfrak{S}_k} \left( \cH^{\otimes k} \right)$ with respect to the orthogonal basis $\{C_1^{\cH},\dots,C_{m^{\cH}}^{\cH}\}$ as
\begin{align*}
	p_{rs}^{t} = \left| \left\{ l \in \left[(d_\cH)^k\right] \;:\; (i,l) \in O_r^{\cH} \;,\; (l,j) \in O_s^{\cH} \right\} \right|,
\end{align*}
where $(i,j) \in O_{t}^{\cH}$. Here, $p_{rs}^t$ does not depend on the choice of $i$ and $j$. Let $E^s,E^r,E^t$ be the matrices defined in Eq.~\eqref{eq:number_pair} for orbits $O_{s}^{\cH}, O_{r}^{\cH}, O_{t}^{\cH}$, respectively. The following proposition implies that $p_{rs}^t$ can be computed in $\poly(k)$ time.

\begin{proposition}[\cite{gijswijt2009block}]
	The numbers $p_{rs}^{t}$ are given by
	\begin{align*}
		p_{rs}^{t}  = \sum_{B}\prod_{x,y=1}^{d_{\cH}} \binom{(E^t)_{x,y}}{B_{x,1,y},\dots,B_{x,d_{\cH},y}} \, ,
	\end{align*}
where the sum runs over all $B \in \mathbb{Z}_{\geq 0}^{d_{\cH} \times d_{\cH} \times d_{\cH}}$ that satisfy $\sum_{z}B_{x,y,z} = (E^r)_{x,y}$, $\sum_{x}B_{x,y,z} = (E^s)_{y,z}$, $\sum_{y}B_{x,y,z} = (E^{t})_{x,z}$ for all $x,y,z \in [d_{\cH}]$ and $\sum_{x,y,z \in [d_{\cH}]}B_{x,y,z} = k$.  
\end{proposition}

Table \ref{table:reduce_dimension} compares the reduction in the size of the matrices for different values of $k$, using both methods of regular $*$-representation and block-diagonal decomposition. The first column contains $\dim\Br{\Lin(X^{\otimes k}\otimes Y^{\otimes k})}$, for $X = Y =\CC^2$ and different values of $k$. The numbers in the second column correspond to the reduced matrix sizes using regular $*$-representation and the third column contains the block sizes in the block-diagonal decomposition. As illustrated by these examples, the size of the variables and the constraint matrices can be significantly reduced by using block-diagonalization. While the reduction obtained by using the regular $*$-representation is not as strong, it has the advantage that it is easy to compute using the explicit formula given in Eq.~\eqref{eq:regular rep}. 
  
\begin{table}[ht]
  	\begin{center}
  			\begin{tabular}{ |c|c|c|c|c| } 
  			\hline
  			$k$ & $\dim \Lin(X^{\otimes k}\otimes Y^{\otimes k}) $&$\dim \End^{\mathfrak{S}_k}\left(X^{\otimes k}\otimes Y^{\otimes k} \right) $& Block sizes \\
  			\hline 
  			$2$ &$256$ &$136$ & $10,6$ \\
  			\hline 
  			$3$ & $4096$&$816$ &  $20,20,4$ \\ 
  			\hline
  			$4$ & $65536$ & $3876$ & $45,35,20,15,1$ \\
  			\hline
  		\end{tabular}
  		\caption{\label{table:reduce_dimension} Dimensions of $\Lin(X^{\otimes k} \otimes Y^{\otimes k})$, $\End^{\mathfrak{S}_k}(X^{\otimes k} \otimes Y^{\otimes k})$, and the block sizes in the block-diagonal form with $X = Y =\CC^2$.}
  	\end{center}
\end{table}

\subsection{Beyond permutation invariance} \label{sec:Zsymmetry}

So far, we have only focused on the permutation symmetries of convex optimization problem~\eqref{sdp:regD_approximation} arising from considering multiple copies of quantum channels. In this section, we discuss how the group symmetries of the underlying channels may be used to further simplify these convex programs. In particular, we show how the symmetries of the channels can be combined with the permutation symmetry and expressed as invariance under the action of a single group. Theorem~\ref{thm:*-isomorphism decomposition} is then used to simplify the programs. 

Let $G$ be a finite group, and denote by $G^k$, the $k$-fold direct product of $G$. Consider the group $H: = G^k \rtimes_\gamma \mathfrak{S}_k$, an outer semi-direct product of $G^k$ and $\mathfrak{S}_k$, defined as follows:
\begin{itemize}
	\item The underlying set is the Cartesian product of the sets $G^k$ and $\mathfrak{S}_k$, i.e., the set of ordered pairs $(g,\pi)$, where $g=\br{g_1,g_2,\ldots,g_k}\in G^k$ and $\pi\in \mathfrak{S}_k$.
	\item $\gamma: \mathfrak{S}_k \rightarrow \mathrm{Aut}\br{G^k}$ is a group homomorphism given by \begin{equation*}
		\gamma(\pi)\br{g_1,g_2,\ldots,g_k}=\br{g_{\pi(1)},g_{\pi(2)},\ldots,g_{\pi(k)}}\enspace,
	\end{equation*}
	for every $\pi\in \mathfrak{S}_k$ and $g=\br{g_1,g_2,\ldots,g_k}\in G^k$.
	\item The group operation $*$ is defined for any pair $(g,\pi),(g',\pi')\in H$ as 
	\begin{equation*}
		(g',\pi')*(g,\pi)=(g' \gamma(\pi')(g),\pi'\pi).
	\end{equation*}
\end{itemize}

Consider an arbitrary action of $G$ on a finite dimensional Hilbert space $\cH$, and the natural action of $\mathfrak{S}_k$ on $\cH^{\otimes k}$ defined for every $\pi \in \mathfrak{S}_k$ as
\begin{align}\label{eq:sk_action}
	\pi \cdot \left( h_1 \otimes \dots \otimes h_k \right) =  h_{\pi^{-1}(1)} \otimes \dots \otimes h_{\pi^{-1}(k)} \quad,\quad h_i \in \cH \,,\, \forall i\in [k] \enspace.
\end{align}  
Then it is easy to check that the following defines an action of $H=G^k \rtimes_\gamma \mathfrak{S}_k$ on $\cH^{\otimes k}$:
\begin{align}\label{eq:H_action}
	(g,\pi) \cdot (h_1 \otimes \dots \otimes h_k )  = g_1\cdot h_{\pi^{-1}(1)} \otimes \dots \otimes g_k \cdot h_{\pi^{-1}(k)} \quad,\quad h_i \in \cH \,,\, \forall i\in [k] \enspace, 
\end{align}
for all $\pi \in \mathfrak{S}_k$ and $g \in G^k$. In particular, we have
\begin{align*}
	(g',\pi') \cdot \br{(g,\pi) \cdot (h_1 \otimes \dots \otimes h_k )}  = \br{(g',\pi') * \br{(g,\pi)} \cdot (h_1 \otimes \dots \otimes h_k )} \enspace,
\end{align*}
for every $(g,\pi),(g',\pi')\in H$.

For $\cH\in\{X,Y\}$, let $\rho_\cH:G\rightarrow \GL(\cH)$ be the representation of $G$ defined by its action on $\cH$ and $\rho_{X\otimes Y} \coloneqq \rho_X\otimes \rho_Y$. Let $\sigma_\cH$ denote the representation of $G^k$ on $\cH^{\otimes k}$ given by $\sigma_\cH(g)\coloneqq \rho_\cH(g_1)\otimes \ldots \otimes \rho_\cH(g_k)$, for every $g\in G^k$. As before, denote by $P_\cH$ the representation of $\mathfrak{S}_k$ on $\cH^{\otimes k}$ defined by the action~\eqref{eq:sk_action}. Then the representation of $H=G^k \rtimes_\gamma \mathfrak{S}_k$ defined above on $\cH^{\otimes k}$ is given by $\sigma_\cH(g) P_\cH(\pi)$, for every $(g,\pi)\in H$. Note that in \eqref{eq:H_action}, for $\cH=X\otimes Y$, the action of $(g, \pi)$ on $(X\otimes Y)^{\otimes k}$ corresponds to the simultaneous permutation of the $X$ and $Y$ tensor factors followed by applying $\rho_{X}(g_i)\otimes \rho_{Y}(g_i)$ on $i$-th $X\otimes Y$ tensor factor. When the subsystems are reordered as $X^{\otimes k} \otimes Y^{\otimes k}$, this action is simply given by $\sigma_X(g)P_X(\pi) \otimes \sigma_Y(g)P_Y(\pi)$.
With the above notation, we are now ready to state the following proposition:

\begin{proposition}\label{prop:H_invariant_SDP}
	Let $\cN_{X \rightarrow Y}$ and $\cM_{X \rightarrow Y}$ be a quantum channels with Choi operators $J^\cN,J^\cM\in \End^G(X\otimes Y)$, for some finite group G. Then the convex program \eqref{sdp:regD_approximation} has an optimal solution $A \in \End^{H}(X^{\otimes k} \otimes Y^{\otimes k})$, where $H = G^k \rtimes_\gamma \mathfrak{S}_k$.
\end{proposition}

\begin{proof}
	The proof is based on convexity and exactly follows the steps of the proof of Lemma~\ref{lem:inv_space_1}, except the group average operator $\bar{A}$ is now obtained with respect to the group $H$.
\end{proof}


Next, we discuss the irreducible representations of $G^k \rtimes_\gamma \mathfrak{S}_k$ and the corresponding multiplicities for the action of $H$ on $\cH^{\otimes k}$, defined in Equation~\eqref{eq:H_action}. First, we need to introduce some notations.

Suppose that $G$ has $t$ irreducible representations and let $m_i$ denote the multiplicity of the $i$-th irreducible representation in the representation $\rho_\cH$ of $G$ on $\cH$. Let $\mathcal{T}(k)$ be the collection of all $t$-tuples $(k_1,\dots,k_t)$  of non-negative integers such that $\sum_{i=1}^{t}k_i = k$. For $(k_1,\dots,k_t) \in \mathcal{T}(k)$ and $(\lambda^1,\dots,\lambda^t)$ satisfying $\lambda^i \vdash_{m_i} k_i$, for every $i\in[t]$, we write $(\lambda^1,\dots,\lambda^t) \vdash_{\bm{m}} (k_1,\dots,k_t)$, where $\bm{m}=\br{m_1,\ldots,m_t}$. We then use a result from~\cite{Polak_thesis}.

\begin{proposition}[Proposition 3.1.1, \cite{Polak_thesis}]
	The irreducible representations of $H = G^k \rtimes_\gamma \mathfrak{S}_k$ are labeled by 
	\begin{align*}
		\{(k_1,\dots,k_t),(\lambda^1,\dots,\lambda^t): (k_1,\dots,k_t) \in \mathcal{T}(k), (\lambda^1,\dots,\lambda^t) \vdash_{\bm{m}} (k_1,\dots,k_t)  \}.
	\end{align*}    
	and the corresponding multiplicities are $\prod_{i=1}^{t}|T_{\lambda^{i},m_i}|$.
\end{proposition} 
Note that $|\mathcal{T}(k)|=\binom{k+t-1}{t-1}$, where $t$, the number of irreducible representations of $G$, is a property of $G$ and independent of $k$. Since $G$ is a finite group, we have $t\leq |G|$. Moreover, for a fixed tuple $(k_1,\dots,k_t) \in \mathcal{T}(k)$, by Inequality~\eqref{eq:partitions}, we have the size of the set $\{(\lambda^1,\dots,\lambda^t): (\lambda^1,\dots,\lambda^t) \vdash (k_1,\dots,k_t) \}$ is at most $\prod_{i=1}^{t}(k_i+1)^{m_i}$. Since $m_i\leq \dim(\cH)$, for every $i\in[t]$, the number of irreducible representations of $H$ is polynomial in $k$. 
Since $|T_{\lambda^{i},m_i}| \leq (k_i+1)^{m_i(m_i-1)/2}$, the multiplicity of the corresponding irreducible representation of $H$ is at most $\prod_{i=1}^{t}|T_{\lambda^{i},m_i}|
\leq \prod_{i=1}^{t}(k_i+1)^{m_i(m_i-1)/2}$.

\subsubsection{Application to the generalized amplitude damping channel}

As an application, we consider the generalized amplitude damping (GAD) channel defined as
\begin{align} \label{eq:GAD_channel}
	\cA_{p,q} (\rho) = \sum_{i=1}^{4}A_i\rho A_{i}^{*}, \quad p,q \in [0,1]
\end{align}  
with the Kraus operators
\begin{equation}
	\begin{aligned}
		A_1 &= \sqrt{1-q}(\ketbra{0}{0}+\sqrt{1-p}\ketbra{1}{1}), &&A_2= \sqrt{p(1-q)}(\ketbra{0}{1}),\\
		A_3 &= \sqrt{q}(\sqrt{1-p}\ketbra{0}{0}+\ketbra{1}{1}),  &&A_4= \sqrt{pq}\ketbra{1}{0}.  
	\end{aligned}
\end{equation}

The GAD channel reduces to the conventional amplitude damping (AD) channel, when $q=0$. In this case we have $X = Y = \CC^2$. Let $N_{p,q}$ be the Choi matrix of $\mathcal{A}_{p,q}$. 
Note that for the Pauli $\mathrm{Z}$ operator given by
\begin{equation*}
	\mathrm{Z} = 
	\begin{pmatrix}
		1 & 0 \\
		0 & -1
	\end{pmatrix},
\end{equation*} 
we have, $(\mathrm{Z}\otimes \mathrm{Z})N_{p,q}(\mathrm{Z} \otimes \mathrm{Z}) = N_{p,q}$ for all $p,q \in [0,1]$. Let $G = \mathbb{Z}_2$ be the cyclic group of order 2 and define the group representation $\rho: G \ra \GL(\CC^2)$ given by $\rho(1) = \mathrm{Z}$. Then for the representation $\rho_{X\otimes Y}$ defined for every $g\in\mathbb{Z}_2$ as $\rho_{X\otimes Y}(g)=\rho(g)\otimes \rho(g)$, we have $N_{p,q}\in \End^G(X\otimes Y)$. The representation $\rho_{X\otimes Y}$ has two irreducible representations, which are both $1$-dimensional (since $G$ is an Abelian group). In this representation, the multiplicities are $(m_1,m_2) = (2,2)$. Therefore, the multiplicities in the representation of $H=G^k \rtimes_\gamma \mathfrak{S}_k$ on $X^{\otimes k} \otimes Y^{\otimes k}$ are at most $(k_1+1)(k_2+1) \leq (k_1+k_2+2)^{2}/4 = (k+2)^2/4$. 
Furthermore, since $t=2$, we have $|\mathcal{T}(k)|=k$, and for any $(k_1,k_2) \in \mathcal{T}(k) $, the size of the set $\{(\lambda^1,\lambda^2): (\lambda^1,\lambda^2) \vdash (k_1,k_2) \}$ is at most $(k_1+1)^2(k_2+1)^2 \leq (k+2)^4/16$. Therefore the number of irreducible representations of $H$ is at most $k(k+2)^{4}/16$. Since the dimension of the invariant subspace is equal to the sum of squares of the multiplicities of the irreducible representations, we have $\dim \End^{H}(X^{\otimes k} \otimes Y^{\otimes k}) \leq \left( (k+2)^2/4\right)^2 k(k+2)^4/16 = k(k+2)^8/256$. 

Therefore, in this example, by considering the additional $\mathrm{Z}$ symmetry discussed above, we can reduce the dimension of the invariant subspace from $\mathrm{O}\br{k^{16}}$ for the permutation action (see Eq.~\eqref{eq:dim_Sk_inv}) to $\mathrm{O}\br{k^9}$, when we combine the two symmetries. Moreover, the maximum block size is reduced from $\mathrm{O}\br{k^6}$ (see Eq.~\eqref{eq:dim_S_lambda}) to $\mathrm{O}\br{k^2}$. This shows the potential of the approach introduced above for channels with stronger symmetries.

In the following table, we compare the dimensions of the $\mathfrak{S}_k$-invariant and $H$-invariant  subspace of operators for $X=Y=\CC^2$ and different values of $k$.  We also list the number of irreducible representations and the maximum block size of the invariant operators in the block-diagonal form.

\begin{table}[h!]
	\centering
	\small
	\begin{tabular}{|c|c|c|c|c|c|c|}
		\hline
		\multirow{3}{*}{$k$} & \multicolumn{3}{c|}{$\mathfrak{S}_k$} & %
		\multicolumn{3}{c|}{$G^k \rtimes \mathfrak{S}_k$}\\
		
		\cline{2-7}
		& $\dim \End^{\mathfrak{S}_k}(\cH^{\otimes k})$ & max.block size &\#-irreps & $\dim \End^{G^k \rtimes \mathfrak{S}_k}(\cH^{\otimes k})$ & max.block size & \#-irreps  \\
		\hline
		$2$&136 &10 & 2&36 &4 &5 \\
		\hline
		$3$& 816& 20& 3&120 &6 &8 \\
		\hline
		$4$& 3876& 45& 5&330 &9 &14 \\
		\hline
		$5$& 15504& 84& 6&792 &12 & 20\\
		\hline
		$6$& 54264& 140&9 &1716 &16&30 \\
		\hline
		$7$& 170544&224 &11 &3432 &20&40 \\
		\hline
		$8$& 490314& 360& 15&6435&25& 55\\
		\hline
		$9$& 1307504&540 & 18&11440&30&70 \\
		\hline
		$10$& 3268760&770 &23 & 19448 &36&91 \\
		\hline
	\end{tabular}
	\label{table:Zsym}
	\caption{The comparison of the reductions obtained by considering invariance under the action of $\mathfrak{S}_k$ and $G^{k}\rtimes \mathfrak{S}_k$ on $\cH^{\otimes k}$, where $\cH = \CC^2 \otimes \CC^2$.}
\end{table}

 We use our method for efficient computation of the $\#$-R\'enyi divergence between multiple copies of channels to provide improved upper bounds on the regularized Umegaki divergence between the AD channel $\cA_{0.3,0}$ and the GAD channel $\cA_{p,0.9}$, over the range $p\in[0.4,0.8]$. Note that the Umegaki divergence between these channels is known to be non-additive~\cite{fang2020chain}, i.e., $\rD^{\mathrm{reg}}(\cA_{0.3,0} \| \cA_{p,0.9}) > \rD(\cA_{0.3,0} \| \cA_{p,0.9})$. Figure \ref{fig:divergence} illustrates the improvement obtained using $\newD_{\alpha}$ on $k=1$ and $k=6$ copies compared to $\widehat{D}_{\alpha}$, for $\alpha=2$. The convex programs are implemented in MATLAB using the CVX package~\cite{cvx} and the CVXQUAD package~\cite{cvxquad}, via the MOSEK solver~\cite{mosek}. Computations in this paper were done using Intel(R) Core i5-6300U with 16GB of RAM memory.  The running time for $k=6$ copies on our program is less than $45$ minutes while the program without using symmetry cannot be carried out due to insufficient memory. We note that without using the symmetry reduction the matrices are of size $4096 \times 4096$.    
\begin{figure}[ht!]
    \centering
%
%
\definecolor{mycolor1}{rgb}{1,0,0}%
\definecolor{mycolor2}{rgb}{0.85000,0.32500,0.09800}%
\definecolor{mycolor3}{rgb}{0.92900,0.69400,0.12500}%
\definecolor{mycolor4}{rgb}{0,0,1}%
\begin{tikzpicture}

\begin{axis}[%
width=3.5in,
height=3in,
at={(0.47in,1.25in)},
scale only axis,
xmin=0.4,
xmax=0.8,
ymin=0.5,
ymax=3.00,
xlabel={$p$},
ylabel style={rotate=-90},
axis background/.style={fill=white},
legend style={at={(axis cs:0.41,1.11)},anchor=south west, font=\small}
]
\addplot [color=mycolor1, line width=1.5pt]
  table[row sep=crcr]{%
0.4	    2.9574\\
0.42	2.9065\\
0.44	2.8612\\
0.46	2.8209\\
0.48	2.7853\\
0.50	2.754\\
0.52	2.7267\\
0.54	2.7031\\
0.56	2.6829\\
0.58	2.6661\\
0.60	2.6524\\
0.62	2.6417\\
0.64	2.6339\\
0.66 	2.6289\\
0.68 	2.6267\\
0.70 	2.6273\\
0.72 	2.6307\\
0.74 	2.637\\
0.76 	2.6463\\
0.78 	2.6586\\
0.80 	2.6743\\
};
\addlegendentry{$\widehat{\rD}_{2}$ and $\rD_{\max}$}

\addplot [color=mycolor2, line width=1.5pt, dashed]
  table[row sep=crcr]{%
0.4	    2.0992\\
0.42	2.089\\
0.44	2.0834\\
0.46	2.0818\\
0.48	2.0839\\
0.50	2.0896\\
0.52	2.0983\\
0.54	2.1101\\
0.56	2.1246\\
0.58	2.1418\\
0.60	2.1614\\
0.62	2.1834\\
0.64	2.2076\\
0.66 	2.2342\\
0.68 	2.2629\\
0.70 	2.2938\\
0.72 	2.3269\\
0.74 	2.3623\\
0.76 	2.4001\\
0.78 	2.4403\\
0.80 	2.4831\\
};
\addlegendentry{$\newD_{2}$ with $k=1$}

\addplot [color=mycolor3, line width=1.5pt, dotted]
  table[row sep=crcr]{%
0.4  1.9017\\
0.42  1.8875\\
0.44  1.8795\\
0.46  1.8801\\
0.48  1.8738\\
0.5  1.8713\\
0.52  1.87\\
0.54  1.8702\\
0.56  1.8739\\
0.58  1.8804\\
0.6  1.8906\\
0.62  1.9036\\
0.64  1.9203\\
0.66  1.9405\\
0.68  1.9655\\
0.7  1.9955\\
0.72  2.0298\\
0.74  2.0671\\
0.76  2.1101\\
0.78  2.1549\\
0.8  2.2029\\
};
\addlegendentry{$\newD_{2}$ with $k=6$}

\addplot [color=mycolor4, line width=1.5pt]
table[row sep=crcr]{%
	0.4	    0.75246\\
	0.42	0.77795\\
	0.44	0.80767\\
	0.46	0.8411\\
	0.48	0.87786\\
	0.50	0.91763\\
	0.52	0.96018\\
	0.54	1.0054\\
	0.56	1.0531\\
	0.58	1.1034\\
	0.60	1.1561\\
	0.62	1.2115\\
	0.64	1.2696\\
	0.66 	1.3305\\
	0.68 	1.3944\\
	0.70    1.4615\\
	0.72 	1.5321\\
	0.74 	1.6065\\
	0.76 	1.6851\\
	0.78 	1.7682\\
	0.80 	1.8565\\
};
\addlegendentry{$\rD$}

\end{axis}
\end{tikzpicture}%
  	\caption{Comparison between different bounds on $\rD^{\mathrm{reg}}(\cA_{0.3,0} \| \cA_{p,0.9})$ over the range $p \in [0.4,0.8]$.} 
  	\label{fig:divergence}
\end{figure}
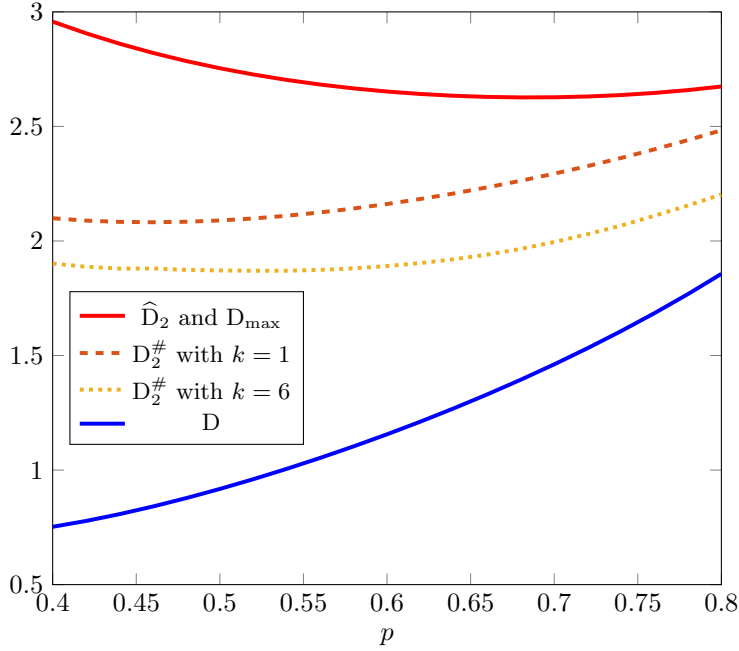

	

\section{Efficient bounds on classical capacity of quantum channels}
\label{section:5}
\suppress{In this section, we provide another example of how symmetries can be utilized in computing a regularized information theoretic quantity.}
The \emph{(unassisted) classical capacity} of a quantum channel is defined as the maximum rate at which classical information can be transmitted over the quantum channel in the asymptotic limit of many channel uses.
For a quantum channel $\cN$, the classical capacity is characterized by the regularized Holevo information~\cite{Holevo_98,Schumacher_97} as
\begin{align*}
    \CCap(\cN) = \lim_{k \to \infty}\frac{1}{k} \chi(\cN^{\otimes k}) \enspace,
\end{align*}
where $\chi(\cN)$ is the Holevo capacity of the channel $\cN$ defined as \begin{align*}
    \chi(\cN)\coloneqq \max_{\cE=\{p_i,\rho_i\}_i} \mathrm{H}\left( \sum_i p_i \rho_i \right)-\sum_i p_i \mathrm{H}\left( \cN(\rho_i) \right) \enspace,
\end{align*} 
where the maximization is over all quantum ensembles $\cE=\{p_i,\rho_i\}_i$. Here, $\mathrm{H}$ denotes the von Neumann entropy, defined as $\mathrm{H}(\sigma)\coloneqq -\tr\br{\rho \log \rho}$, for every positive semidefinite operator $\rho$. Note that the Holevo information is in general non-additive~\cite{hastings2009superadditivity}.

We denote by $\cV_\mathrm{cb}(X,Y)$ the set of constant bounded subchannels from $\Lin(X)$ to $\Lin(Y)$ defined as 
\begin{align*}
    \cV_\mathrm{cb}(X,Y) \coloneqq \left\{ \cM \in \CP(X:Y): \exists\, \sigma \in \D(Y) \text{ s.t. } \cM_{X\to Y}(\rho)\leq \sigma, \forall \rho \in \D(X)\right\} \enspace.
\end{align*}
Let $\cV(X,Y) \coloneqq \{\cM \in \CP(X:Y): \beta(J_{XY}^{\cM})\leq 1\}$, with $\beta(J_{XY}^{\cM})$ defined in terms of the following SDP
\begin{align*}
    \beta(J_{XY}^{\cM}) \coloneqq \min_{R_{XY},S_Y} \tr(S_Y) \text{\quad s.t. \quad} R_{XY} \pm (J_{XY}^{\cM})^{T_Y} \geq 0 \;,\; \id_{X} \otimes S_Y \pm R_{XY}^{T_Y} \geq 0 \enspace,
\end{align*}
where $(\cdot)^{T_Y}$ denotes the partial transpose on system $Y$. Note that the set $\cV(X,Y)$ is a convex subset of $\cV_\mathrm{cb}(X,Y)$ containing all the constant channels~\cite{XWang_19}.

Let $\textbf{D}$ be a generalized quantum divergence. For any quantum channel $\cN_{X \to Y}$, define
\begin{align*}
\Upsilon(\textbf{D},k)(\cN): = \min_{\cM \in \cV(X^{\otimes k},Y^{\otimes k})} \textbf{D}(\cN^{\otimes k}\| \cM).
\end{align*}

The following proposition provides upper bounds on the classical capacity of a quantum channel.  

\begin{proposition} [\cite{XWang_19}] \label{prop:upperbound_capacity}
	Let $\rDbf$ be a generalized quantum divergence. If $\rDbf$ is bounded below by the Umegaki relative entropy on quantum states and the corresponding channel divergence is subadditive under tensor product of channels, then, for any $k \geq 1$,
	\begin{align*}
	\CCap(\cN) \leq \frac{1}{k} \Upsilon(\rDbf,k)(\cN) \, .
	\end{align*}
\end{proposition}

\begin{proof}
    The proof can be found in \cite{XWang_19}, but we include a concise proof for the reader's convenience. 
    As shown in~\cite{ohya1997capacities} the Holevo information can be written as a divergence radius:
 	\begin{align*}
 		\chi(\cN) &= \min_{\sigma \in \D(Y)} \max_{\rho \in  \D(X)} \rD(\cN(\rho)\|\sigma) \\
 		&= \min_{\cM \in \cV_{\mathrm{cb}}(X,Y)} \max_{\rho \in \D(X)} \rD(\cN(\rho)\|\cM(\rho)) \\
 		&\leq \min_{\cM \in \cV(X,Y)} \max_{\rho \in \D(X)} \rD(\cN(\rho)\| \cM(\rho)) \\
 		&\leq \min_{\cM \in \cV(X,Y)} \rD(\cN \| \cM)
 		\end{align*}
 		where we used the fact that if $\sigma \leq \sigma'$ then $\rD(\rho \| \sigma) \geq \rD(\rho \| \sigma')$ and the fact that  $\cV(X,Y)\subseteq \cV_{\mathrm{cb}}(X,Y)$. So, for $n,k \in \mathbb{N}$, we have
 		\begin{align*}
 		\chi(\cN^{\otimes nk}) &\leq \min_{\cM \in \cV(X^{\otimes nk}, Y^{\otimes nk})} \rD(\cN^{\otimes nk}\| \cM) \\
 		&\leq \min_{\cM \in \cV(X^{\otimes k}, Y^{\otimes k})} \rD(\cN^{\otimes nk}\| \cM^{\otimes n}),
 	\end{align*}
 	where we used the fact that if $\cM \in \cV(X^{\otimes k}, Y^{\otimes k})$, then $\cM^{\otimes n} \in \cV(X^{\otimes nk},Y^{\otimes nk})$. Since $\rDbf$ is bounded below by $\rD$ and subadditive under tensor product of channels, we have
 	\begin{align*}
 		\frac{1}{nk}\chi(\cN^{\otimes nk}) \leq \min_{\cM \in \cV(X^{\otimes k}, Y^{\otimes k})} \frac{1}{k} \; \rDbf(\cN^{\otimes k} \| \cM) = \frac{1}{k}\Upsilon(\rDbf,k)(\cN). 
 	\end{align*}
 	Taking the limit as $n \to \infty$, we get the desired result.
\end{proof}
 
Note that by Proposition~\ref{prop:compair_divergences}, for $\alpha\in (1,2]$, we have
\begin{align*}
    \Upsilon(\widetilde{\rD}_{\alpha},k)(\cN) \leq \Upsilon(\newD_{\alpha},k)(\cN) \leq \Upsilon(\widehat{\rD}_{\alpha},k)(\cN) \leq \Upsilon(\rD_{\max},k)(\cN).
\end{align*}
\begin{remark}
	If in addition the generalized quantum divergence $\mathbf{D}$ satisfies $\widetilde{\rD}_{\alpha} \leq \mathbf{D}$, for some $\alpha \in (1,\infty)$, then $\frac{1}{k}\Upsilon(\textbf{D},k)(\cN)$ is a strong converse bound, i.e., above this communication
	rate, the error probability goes to 1.  
\end{remark}

Both $\rD_{\max}$ and $\widehat{\rD}_{\alpha}$ have the desired properties and were used in~\cite{XWang_19} and~\cite{Fang_19} to obtain bounds on the classical capacity. On the other hand, $\widetilde{\rD}_{\alpha}$ is not always additive~\cite{fang2020chain} so it cannot be used in general.  The best-known general strong converse bound is given by $\frac{1}{k}\Upsilon(\widehat{\rD}_{\alpha},k)$, and it is SDP computable~\cite{Fang_19}. 
For $\textbf{D} = \newD_{\alpha}$, using the formulation of the channel divergence given in Eqs.~\eqref{eq:sdp Dsharp 1} and~\eqref{eq:sdp Dsharp 2}, the converse bound of Proposition~\ref{prop:upperbound_capacity} can be written in terms of a convex program. For every $k\geq 1$, we have
\begin{equation} \label{sdp:Upsilon}
    \begin{aligned}
        \Upsilon(\newD_{\alpha},k)(\cN) = \frac{1}{\alpha-1} \log \quad \min \quad &    \| \tr_{Y^{\otimes k}}(A) \|_{\infty}\\
        \textrm{s.t.} \quad & J^{\cN^{\otimes k}}  \leq J^{\cM} \#_{1/\alpha} A \enspace,\\
        &R \pm (J^{\cM})^{T_{Y^{\otimes k}}} \geq 0 \enspace,\\
        &(I_{X^{\otimes k}}\otimes S) \pm R^{T_{Y^{\otimes k}}} \geq 0 \enspace,\\
        &\tr(S) \leq 1 \enspace,\\
        &A,J^{\cM},R \in \Pos(X^{\otimes k}\otimes Y^{\otimes k}) \,,\, S\in \Pos(Y^{\otimes k})\enspace. 
    \end{aligned}
\end{equation}
     
Note that the optimization problem in Eq.~\eqref{sdp:Upsilon} does not scale well with $k$ since the sizes of the constraint matrices grow exponentially fast. This bottleneck will be addressed in the next section.


\subsection{Exploiting symmetries to simplify the problem}

Using a similar argument as in Lemma~\ref{lem:inv_space_1}, one may restrict the feasible region of the convex program~\eqref{sdp:Upsilon} to the $\mathfrak{S}_k$-invariant subspace of operators.

\begin{lemma} \label{lem:inv_space_2}
	For every $\alpha \in (1,\infty)$, the convex program \eqref{sdp:Upsilon} has an optimal solution $(A,R,J^\cM,S)$, with $A,R,J^{\cM} \in  \End^{\mathfrak{S}_k}\left(X^{\otimes k}\otimes Y^{\otimes k} \right)$ and $S \in \End^{\mathfrak{S}_k}(Y^{\otimes k})$. 
\end{lemma}

\begin{proof}
    It is straightforward to check that by Slater's condition the optimal value is achieved by a feasible solution. For an arbitrary feasible solution $(A,J^\cM,R,S)$, we will prove that the corresponding group-average operators $(\overline{A},\overline{J^\cM},\overline{R},\overline{S})$ are feasible with an objective value not greater than the original value. 
  
    For brevity of notation, we write $\Pi(\pi)\coloneqq P_{X\otimes Y}(\pi)$. The first constraint, $J^{\cN^{\otimes k}} \leq \overline{J^\cM}\; \#_{1/\alpha} \;\overline{A}$, follows from a similar argument as in Lemma~\ref{lem:inv_space_1}. For the second constraint note that, for every $\pi \in \mathfrak{S}_k$, $\Pi(\pi)^*=\Pi(\pi)^T$, and we have 
    \begin{align}
        \left(\Pi(\pi) J^{\cM} \Pi(\pi)^{*} \right)^{T_{Y^{\otimes k}}} = \left(\Pi(\pi) J^{\cM} \Pi(\pi)^{T} \right)^{T_{Y^{\otimes k}}} = \Pi(\pi)(J^{\cM})^{T_{Y^{\otimes k}}} \Pi(\pi)^{T} \enspace.
    \end{align}
    Therefore,
    \begin{align}
        \left( \overline{J^\cM} \right)^{T_{Y^{\otimes k}}} &= \left( \frac{1}{|\mathfrak{S}_k|}\sum_{\pi \in \mathfrak{S}_k}\Pi(\pi)J^{\cM}\Pi(\pi)^{*} \right)^{T_{Y^{\otimes k}}} = \frac{1}{|\mathfrak{S}_k|}\sum_{\pi \in \mathfrak{S}_k}\Pi(\pi)(J^{\cM})^{T_{Y^{\otimes k}}}\Pi(\pi)^{*} \enspace,
    \end{align}
    and the feasibility of $J^{\cM}$ and $R$ implies $-\overline{R} \leq \left( \overline{J^\cM} \right)^{T_{Y^{\otimes k}}} \leq \overline{R}$. Similarly, we get
    \begin{align}
        \left( \overline{R} \right)^{T_{Y^{\otimes k}}} &= \left( \frac{1}{|\mathfrak{S}_k|}\sum_{\pi \in \mathfrak{S}_k}\Pi(\pi) \,R\,\Pi(\pi)^{*} \right)^{T_{Y^{\otimes k}}} =  \frac{1}{|\mathfrak{S}_k|}\sum_{\pi \in \mathfrak{S}_k}\Pi(\pi)(R)^{T_{Y^{\otimes k}}}\Pi(\pi)^{*} \enspace,
    \end{align}
    and the feasibility of $S$ and $R$ implies $-\id_{X^{\otimes k}} \otimes \overline{S} \leq \left( \overline{R} \right)^{T_{Y^{\otimes k}}} \leq \id_{X^{\otimes k}} \otimes \overline{S}$. Finally, the forth constraint holds since $\tr(\overline{S})=\tr(S)\leq 1$.

    For the objective function, using the same argument as in Lemma~\ref{lem:inv_space_1}, we get $\| \tr_{Y^{\otimes k}} \left( \overline{A} \right) \|_{\infty} \leq \|\tr_{Y^{\otimes k}}(A)\|_{\infty}$. This concludes the proof.
\end{proof}

Next, we show that the convex program~\eqref{sdp:Upsilon} may be reformulated so that it scales only polynomially with $k$.

\begin{theorem} \label{thm:computing_Upsilon_1}
    Let $\cN_{X \rightarrow Y}$ be a quantum channel. For every $k\geq 1$, the strong converse bound $\frac{1}{k}\Upsilon(\newD_{\alpha},k)(\cN)$ of Proposition~\ref{prop:upperbound_capacity} can be formulated as a convex program with only $\mathrm{O}\!\br{ k^{d^2} }$ variables and $\mathrm{O}\!\br{k^d}$ PSD constraints involving matrices of size at most $(k+1)^{d(d-1)/2}$, where $d=d_X d_Y$.
\end{theorem}
\begin{proof}
    Let $Q$ denote the permutation matrix which maps $X^{\otimes k}\otimes Y^{\otimes k}$ to $\br{X\otimes Y}^{\otimes k}$. Then, by Lemma~\ref{lem:inv_space_2}, the optimization problem~\eqref{sdp:Upsilon} can be written as  
  	\begin{align}
  	    \Upsilon(\newD_{\alpha},k)(\cN) = \frac{1}{\alpha-1}\;\log\; &\min \; y \\
  	    \textrm{s.t.} \quad &\tr_{Y^{\otimes k}}(A) \leq y \; \id_{X^{\otimes k}} \enspace, \label{constraint_1}\\ 
  	    &\br{J^{\cN}}^{\otimes k} \leq J^{\cM} \#_{1/\alpha} A \enspace, \label{constraint_2} \\
  	    &R \pm (J^{\cM})^{T_{Y^{\otimes k}}} \geq 0 \enspace, \label{constraint_3}\\
  	    &Q(\id_{X^{\otimes k}}\otimes S)Q^T \pm R^{T_{Y^{\otimes k}}} \geq 0 \enspace, \label{constraint_4}\\
  	    &\tr(S) \leq 1 \enspace, \label{constraint_5}
  	\end{align}
    where $A,J^{\cM},R \in \End^{\mathfrak{S}_k}\left(\br{X \otimes Y}^{\otimes k} \right)$ and $S\in \End^{\mathfrak{S}_k}\left(Y^{\otimes k} \right)$ are positive semidefinite operators and $y\in \RR$.

    Following the notation introduced in Theorem~\ref{thm: computing_channel_divergence}, for $\cH\in\{X,Y,X\otimes Y\}$, let $\phi_\cH: \End^{\mathfrak{S}_k}(\cH^{\otimes k}) \to \bigoplus_{i=1}^{t^\cH} \CC^{m_i^\cH \times m_i^\cH}$ be the bijective linear map which block-diagonalizes the corresponding invariant algebra, where to simplify the notation, the blocks are indexed by $i\in[t^\cH]$ instead of $\lambda\in \mathrm{Par}(d_\cH,k)$. For $Z\in \End^{\mathfrak{S}_k} (\cH^{\otimes k})$, we write $\llbracket \phi_\cH(Z) \rrbracket_i$ to denote the $i$-th block of $\phi_\cH(Z)$. Since $J^{\cN^{\otimes k}}$, $J^\cM$, $A$ and $R$, $(J^{\cM})^{T_{Y^{\otimes k}}}$ are elements of $\End^{\mathfrak{S}_k}\left(\br{X\otimes Y}^{\otimes k}\right)$, the constraints~\eqref{constraint_2} and~\eqref{constraint_3} can be mapped into the direct sum form under $\phi_{X \otimes Y}$. Similarly, since $Q\br{\id_{X^{\otimes k}}\otimes S}Q^T$, $R^{T_{Y^{\otimes k}}} \in \End^{\mathfrak{S}_k}\left(\br{X\otimes Y}^{\otimes k}\right)$, by properties~$2$ and~$5$ of the $\alpha$-geometric mean, the constraint~\eqref{constraint_4} can be decomposed into constraints involving the smaller diagonal blocks  by applying $\phi_{X\otimes Y}$. Finally, since $\tr_{Y^{\otimes k}}(A)$, $\id_{X^{\otimes k}} \in \End^{\mathfrak{S}_k}(X^{\otimes k})$, the constraint~\eqref{constraint_1} can be mapped by $\phi_X$ into the direct sum form. The transformed convex program is given by
  	\begin{IEEEeqnarray*}{rCl"r}
  	    \frac{1}{\alpha - 1}\,\log \quad \min& \quad &y \\
  	    \mathrm{s.t.}&  &\left\llbracket \br{\phi_X \circ \tr_{Y^{\otimes k}} \circ \phi_{X\otimes Y}^{-1}} \br{ \oplus_l A_l} \right\rrbracket_j \leq y\,\id_{m_j^X} \enspace,\\
  	    & &\left\llbracket \phi_{X\otimes Y} \br{\br{J^{\cN}}^{\otimes k}} \right\rrbracket_i \leq  J_i \, \#_{1/\alpha} \, A_i \enspace,\\
  	    & &R_i \pm \left\llbracket \br{ \phi_{X\otimes Y} \circ T_{Y^{\otimes k}} \circ \phi_{X\otimes Y}^{-1}}\br{ \oplus_l A_l} \right\rrbracket_i \geq 0 \enspace,\\
  	    & &\left\llbracket \phi_{X\otimes Y} \br{ Q\br{\id_{X^{\otimes k}}\otimes \phi_Y^{-1}\br{ \oplus_r S_r}} Q^T} \right\rrbracket_i \pm \left\llbracket \br{ \phi_{X\otimes Y} \circ T_{Y^{\otimes k}} \circ \phi_{X\otimes Y}^{-1}}\br{ \oplus_l R_l} \right\rrbracket_i \geq 0\enspace,\\
  	    & &\textstyle \sum_r \tr(S_r) \leq 1\enspace,\\
  	    & &A_i,R_i,J_i \in \Pos\br{\CC^{m_i^{X\otimes Y}}}, S_r \in \Pos\br{\CC^{m_r^Y}}\enspace,
  	\end{IEEEeqnarray*}
    for all $i \in \Br{t^{X\otimes Y}}$, $j \in \Br{t^X}$, and $r \in \Br{t^Y}$. 
    The statement of the theorem follows since for $\cH\in\{X,Y,X\otimes Y\}$, we have $t^\cH \leq (k+1)^{d_\cH}$ and $m_i^\cH \leq (k+1)^{d_\cH(d_\cH-1)/2}$, for every $i\in \Br{t^\cH}$.
 \end{proof}

Finally, we show how to efficiently compute a formulation of  $\frac{1}{k}\Upsilon(\newD_{\alpha},k)(\cN)$ as a convex program of polynomial size.

\begin{theorem} \label{thm:computing_Upsilon}
    Let $\cN_{X \rightarrow Y}$ be a quantum channel. There exists an algorithm which given as input $J^\cN$ and $k\in \NN$, outputs in $\poly(k)$ time (for fixed $\dim(X \otimes Y)$) the description of a convex program of size described in Theorem~\ref{thm:computing_Upsilon_1} for computing the strong converse bound $\frac{1}{k}\Upsilon(\newD_{\alpha},k)(\cN)$. 
\end{theorem}
\begin{proof}
    As in the proof of Theorem~\ref{thm: computing_channel_divergence}, for $\cH\in\{X,Y,X\otimes Y\}$, let $\cbr{O_{r}^{\cH}}_{r\in[m^\cH]}$ denote the set of orbits of pairs and $\cbr{C_{r}^{\cH}}_{r\in[m^\cH]}$ denote the canonical basis of $\End^{\mathfrak{S}_k}\left(\cH^{\otimes k} \right)$ defined in Eq.~\eqref{eq:canonical_basis}. For every $r\in [m^{X \otimes Y}]$, we define $D_r\coloneqq \tr_{Y^{\otimes k}}\br{C_r^{X \otimes Y}}$. Note that $D_r\in \End^{\mathfrak{S}_k} \br{X^{\otimes k}}$. Then by Theorem~\ref{thm:computing_Upsilon_1}, $\Upsilon(\newD_{\alpha},k)(\cN)$ can be formulated as the following convex program:
    
    \begin{IEEEeqnarray*}{rCl"r}
  	    \frac{1}{\alpha - 1}\,\log \; \min& \;\; &y \\
  	    \mathrm{s.t.}&  &\sum_{r=1}^{m^{X\otimes Y}}z_r \left\llbracket \phi_X\br{D_r}  \right\rrbracket_j \leq y\,\id_{m_j^X} \enspace,\\
  	    & &\left\llbracket \phi_{X\otimes Y} \br{\br{J^{\cN}}^{\otimes k}} \right\rrbracket_i \leq  \sum_{l=1}^{m^{X\otimes Y}}x_l \left\llbracket \phi_{X\otimes Y} \br{C_l^{X\otimes Y}} \right\rrbracket_i \, \#_{1/\alpha} \, \sum_{r=1}^{m^{X\otimes Y}}z_r \left\llbracket \phi_{X\otimes Y} \br{C_r^{X\otimes Y}} \right\rrbracket_i \enspace,\\
  	    & & \sum_{l=1}^{m^{X\otimes Y}}y_l \left\llbracket \phi_{X\otimes Y} \br{C_l^{X\otimes Y}} \right\rrbracket_i \pm \sum_{r=1}^{m^{X\otimes Y}}z_r \left\llbracket  \phi_{X\otimes Y}\br{ \br{C_r^{X\otimes Y}}^{T_{Y^{\otimes k}}}} \right\rrbracket_i \geq 0 \enspace,\\
  	    & &\sum_{s=1}^{m^{Y}}w_s \left\llbracket \phi_{X\otimes Y} \br{ Q\br{\id_{X^{\otimes k}}\otimes C_s^{Y}}Q^T}  \right\rrbracket_i \pm \sum_{l=1}^{m^{X\otimes Y}}y_l \left\llbracket \phi_{X\otimes Y} \br{\br{C_l^{X\otimes Y}}^{T_{Y^{\otimes k}}}} \right\rrbracket_i \geq 0\enspace,\\
  	    & &\sum_{s=1}^{m^{Y}} w_s \, \tr(C_s^{Y}) \leq 1\enspace,\\
  	    & &\sum_{r=1}^{m^{X\otimes Y}}z_r \left\llbracket \phi_{X\otimes Y} \br{C_r^{X\otimes Y}} \right\rrbracket_i \geq 0 \enspace,\\
  	    & &\sum_{r=1}^{m^{X\otimes Y}}y_r \left\llbracket \phi_{X\otimes Y} \br{C_r^{X\otimes Y}} \right\rrbracket_i \geq 0 \enspace,\\
  	    & &\sum_{r=1}^{m^{X\otimes Y}}x_r \left\llbracket \phi_{X\otimes Y} \br{C_r^{X\otimes Y}} \right\rrbracket_i \geq 0 \enspace,\\
  	    & &\sum_{s=1}^{m^{Y}}w_s \left\llbracket \phi_{Y} \br{C_s^{Y}} \right\rrbracket_t \geq 0 \enspace,\\
  	    & &x_r,y_r,z_r,w_s,y\in \RR, \quad \forall r\in [m^{X\otimes Y}],s\in [m^Y]\enspace,
  	\end{IEEEeqnarray*}
    where $j\in \Br{t^X}$, $i\in\Br{t^{X \otimes Y}}$ and $t\in\Br{t^Y}$.    
    
    In Theorem~\ref{thm: computing_channel_divergence}, we showed how to efficiently compute $\phi_{X}(D_r)$, $\phi_{X \otimes Y}(C_{r}^{X \otimes Y})$, and $\phi_{X \otimes Y}(\br{J^{\cN}}^{\otimes k})$. Note that $\phi_{Y}(C_s^Y)$ can be similarly computed in $\poly(k)$ time. Therefore, to complete the proof it suffices to show that $\phi_{X \otimes Y}\br{\br{C_r^{X\otimes Y}}^{T_{Y^{\otimes k}}}}$, $\phi_{X \otimes Y}\br{Q^{T}(I_{X^{\otimes k}} \otimes C_{r}^{Y})Q}$, and $\tr(C_{s}^Y)$ can computed in $\poly(k)$ time. 
    
    Recall that, for every $r \in \Br{m^{X \otimes Y}}$, we have
  	\begin{align*}
  	  	C_r^{X\otimes Y}   = \sum_{(i,j) \in O_r^{X \otimes Y}} \ketbra{i}{j} \, ,
  	\end{align*}
  	where $i= \br{i_{1}^{X}i_{1}^{Y}\cdots i_{k}^{X}i_{k}^{Y}}$ and $j= \br{j_{1}^{X}j_{1}^{Y}\cdots j_{k}^{X}j_{k}^{Y}}$. Therefore, we have
    \begin{align*}
        \br{C_r^{X\otimes Y}}^{T_{Y^{\otimes k}}} = \sum_{(i,j) \in O_r^{X \otimes Y}} \ketbra{i_{1}^{X}j_{1}^{Y}\cdots i_{k}^{X}j_{k}^{Y}}{j_{1}^{X}i_{1}^{Y}\cdots j_{k}^{X}i_{k}^{Y}}= C_{T(r)}^{X\otimes Y}\enspace,
    \end{align*}
    where $T(r)$ denotes the index of the orbit given by 
    \begin{equation*}
        O_{T(r)}^{X \otimes Y}=\cbr{\br{i_{1}^{X}j_{1}^{Y}\cdots i_{k}^{X}j_{k}^{Y},j_{1}^{X}i_{1}^{Y}\cdots j_{k}^{X}i_{k}^{Y}}\,:\,(i,j)\in O_r^{X \otimes Y}}\enspace.
    \end{equation*}
    Therefore, $\phi_{X \otimes Y}\br{\br{C_r^{X\otimes Y}}^{T_{Y^{\otimes k}}}}=\phi_{X \otimes Y}\br{C_{T(r)}^{X\otimes Y}}$ can be computed efficiently.
    
    For $r=1,\dots,m^{X \otimes Y}$, let $(i,j)$ be an arbitrary representative element of $O_{r}^{X \otimes Y}$. Let
    \begin{align*}
    	\alpha_{r}: = (\id_{X^{\otimes k}})_{(i^X,j^X)} \cdot (C_{r}^{Y})_{(i^Y,j^Y)} \, ,
    \end{align*}
    where $i^X=\br{i_1^X \ldots i_k^X}$, $i^Y=\br{i_1^Y \ldots i_k^Y}$, and $j^X$ and $j^Y$ are defined in a similar way. Then we have $Q^{T}(I_{X^{\otimes k}} \otimes C_{r}^{Y})Q = \sum_{r=1}^{m^{X \otimes Y}} \alpha_{r}C_{r}^{X \otimes Y}$, which implies that $\phi_{X \otimes Y}\br{Q^{T}(I_{X^{\otimes k}} \otimes C_{r}^{Y})Q}$  can be computed in $\poly(k)$ time by Lemma~\ref{lem:computing_entries_block}. 
    
    Finally, for every $s \in [m^Y]$, we have 
    \begin{align*}
    	C_{s}^{Y} = \sum_{(i_1\dots i_k,j_{1}\dots j_k) \in O_{s}^{Y}} \ketbra{i_{1}\dots i_{k}}{j_1 \dots j_k} \, .
    \end{align*}
    Therefore, $\tr(C_{s}^{Y})>0$ iff $O_{s}^{Y}=\cbr{(\pi(i),\pi(i)): \pi\in \mathfrak{S}_k}$, for some $i\in[d_Y]^k$. Let $s \in [m^Y]$ such that $\tr(C_{s}^{Y})>0$ and let $(i_1\dots i_k,i_1\dots i_k)$ be an arbitrary representative element of $O_{s}^Y$. For every $a \in [d_Y]$, define $\beta(a) \coloneqq |\{v \in [k]: i_{v} = a \}|$, then $\tr(C_{s}^Y) = k!/\prod_{a \in [d_Y]}\beta(a)!$. 
 \end{proof}

As an example, $\Upsilon(\newD_{2},6)$ is computed for the amplitude damping (AD) channel $\cA_{p,0}$, defined in Eq.~\eqref{eq:GAD_channel}, for different values of $p$. For this channel, the best previously known upper bound on the classical capacity $\CCap(\cA_{p,0})$ for $p \in [0,0.75]$ is given by quantity $\CCap_{\beta}(\cA_{p,0}) = \log(1+\sqrt{1-p})$ in~\cite{XWang_17}.   
Table~\ref{table:upperbound_capacity} shows that $\frac{1}{6}\Upsilon(\newD_{2},6)$ is a slightly improved upper bound compared to the bounds obtained using $\widehat{\rD}_{\alpha}$ and $\rD_\mathrm{\max}$ which happen to coincide for the AD channel~\cite{Fang_19} with the value $\log(1+\sqrt{1-p})$. We remark that the best known upper bound for the AD channel $\cA_{p,0}$ with $p \in [0.75,1]$ is given by the entanglement-assisted classical capacity~\cite{Bennett} of the channel.
  
\begin{table}[ht]
	\centering
	\begin{tabular}{|c|c|c|} 
		\hline 
		$p$ & $\Upsilon(\rD_\mathrm{\max},1)$, $\Upsilon(\widehat{\rD}_2,1)$  and $\CCap_{\beta}$ & $\frac{1}{6}\Upsilon(\newD_{2},6)$ \\
		\hline
		$0.1$ & $0.9626$&$ \textbf{0.9615}$ \\ \hline 
		$0.2$ & $0.9218$&$ \textbf{0.9201}$ \\ \hline 
		$0.3$ & $0.8770$&$ \textbf{0.8745}$ \\ \hline 
		$0.4$ & $0.8274$&$ \textbf{0.8239}$ \\ \hline 
		$0.5$ & $0.7716$&$ \textbf{0.7670}$ \\ \hline
		$0.6$ & $0.7071$&$ \textbf{0.7014}$ \\ \hline  
		$0.7$ & $0.6302$&$ \textbf{0.6234}$ \\ \hline  
		$0.75$ & $0.5850$&$ \textbf{0.5777}$ \\ \hline     
	\end{tabular}
	\caption{\label{table:upperbound_capacity} Upper bounds on the classical capacity of the amplitude damping channel $\cA_{p,0}$ with different parameters $p$.}
\end{table} 
 
\section{Two-way assisted quantum capacity}
 \label{section:6}
 In this section, we consider $\newD_{\alpha}$ in the framework of \emph{generalized Theta-information} which was introduced in~\cite{Fang_19}. As we will see, the generalized Theta-information induced by $\#$-channel divergence gives efficiently computable strong converse bounds on the two-way-assisted quantum capacity, $\QCapTW(\cN)$, for any quantum channel $\cN$.
 \par 
The two-way assisted quantum capacity of a quantum channel $\cN$ is the maximum rate at which quantum information can be transmitted reliably from a sender to a receiver, when the parties are allowed to perform arbitrary LOCC (short for local operations and classical communication) between consecutive channel uses~\cite{bennett1996mixed}. While the two-way assisted quantum capacity for some specific channels such as the quantum erasure channel is known~\cite{bennett1997capacities}, no general characterization of $\QCapTW(\cN)$ is known for an arbitrary quantum channel $\cN$.  
 \par 
 In~\cite{rains1999bound,rains2001semidefinite}, the authors relaxed the set LOCC to a larger class of operations known as PPT-preserving operations, which is the set of channels that are positive partial transpose preserving. A quantum channel $\mathcal{P}_{AB\rightarrow A'B'}$ is PPT-preserving if the linear map $T_{B'}\circ \mathcal{P}_{AB\rightarrow A'B'}\circ T_B$ is completely positive and trace-preserving~\cite{rains2001semidefinite}, where $T_{B'}$ and $T_B$ denote the partial transpose map.  
 For any quantum channel $\cN$, we denote by $\QCapPPT(\cN)$ the PPT-assisted quantum capacity of $\cN$ . In this case, the operations between the channel uses are allowed to be PPT-preserving operations. Because of the  containment LOCC $\subset$ PPT~\cite{rains2001semidefinite}, we have the following inequality
 \begin{align*}
 	\QCapTW(\cN) \leq \QCapPPT(\cN) \, ,
 \end{align*}    
 for all quantum channels $\cN$. 
 
 Inspired by the formulation of the Rains set~\cite{rains2001semidefinite}, in~\cite{Fang_19} the authors introduced the set of subchannels given by the zero set of the Holevo-Werner bound~\cite{holevo2001evaluating} as
 \begin{align*}
 	\Theta(X,Y) \coloneqq \{ \cM \in \CP(X:Y): \exists R_{XY} \text{ s.t. } R_{XY} \pm (J_{XY}^{\cM})^{T_Y} \geq 0, \tr_{Y}(R_{XY}) \leq \id_{X} \} \, .
 \end{align*} 
 
Let  $\textbf{D}$ be a generalized divergence. For any quantum channel $\cN_{X \to Y}$, define
\begin{align*}
	R_{\Theta}(\textbf{D},k)(\cN) \coloneqq \textbf{D}(\cN^{\otimes k}\| \cM^{\otimes k}) \, ,
\end{align*}
where $\cM  = \argmin_{\cM \in \Theta(X,Y)} \textbf{D}(\cN \| \cM)$.       
 
 For any quantum channel $\cN$, by~\cite[Theorem 17]{Fang_19},~\cite[Proposition 5.9]{Omar_20},~\cite[Corollary 5]{berta2018amortization} and the relation between the divergences in Proposition~\ref{prop:compair_divergences}, the following holds:
 \begin{proposition}
 	Let $\cN$ be a quantum channel. For any $\alpha\in (1,2]$ and $k \geq 1$, 
 	\begin{align*}
 		\QCapTW(\cN) \leq \QCapPPT(\cN) \leq \SQCapPPT(\cN) \leq  \frac{1}{k}R_{\Theta}(\newD_{\alpha},k)(\cN) \leq R_{\Theta}(\widehat{\rD}_{\alpha},1)(\cN) \leq R_{\Theta}(\rD_{\max},1)(\cN) \, ,
 	\end{align*}
  where $\SQCapPPT(\cN)$ is the strong converse capacity corresponding to $\QCapPPT(\cN)$. 
 \end{proposition}
 The squashed entanglement of the channel $\cN$ introduced in~\cite{takeoka2014squashed} is known to be a converse bounds for $\QCapPPT(\cN)$. However, it remains open whether it is a strong converse and the quantity itself is NP-hard to compute ~\cite{huang2014computing}. Using a similar method as in Section~\ref{section:5}, we can show that $R_{\Theta}(\newD_{\alpha},k)(\cN)$ can be computed in $\poly(k)$ time for any quantum channel $\cN$. 
 
 As an example, $R_{\Theta}(\newD_{\alpha},6)$ is computed for the qubit amplitude damping channel $\cA_{p,0}$, defined in Eq.~\eqref{eq:GAD_channel}, for values of $p \in [0,1]$. The comparison between the two-way/PPT assisted quantum capacity is given in Figure~\ref{fig:quantum_capacity}. The bound $\frac{1}{6}R_{\Theta}(\newD_{2},6)$ demonstrates an improvement compared to the best previously known strong converse bound given by $R_{\Theta}(\widehat{D}_{2},1)$. 
 
 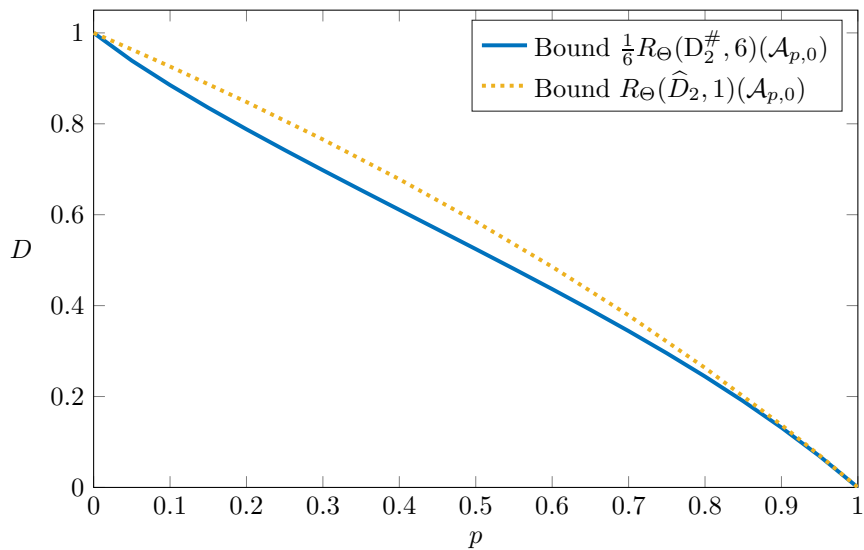
\begin{figure}[ht!]
 	\centering
%
%
\definecolor{mycolor1}{rgb}{0.00000,0.44700,0.74100}%
\definecolor{mycolor2}{rgb}{0.85000,0.32500,0.09800}%
\definecolor{mycolor3}{rgb}{0.92900,0.69400,0.12500}%
\begin{tikzpicture}

\begin{axis}[%
width=4in,
height=2.5in,
at={(1.048in,0.726in)},
scale only axis,
xmin=0,
xmax=1,
ymin=0,
ymax=1.05,
xlabel={$p$},
ylabel style={rotate=-90},
ylabel={$D$},
axis background/.style={fill=white},
legend style={legend cell align=left, align=left, draw=white!15!black}
]

\addplot [color=mycolor1, line width=1.5pt]
  table[row sep=crcr]{%
0  1\\
0.05  0.9389\\
0.1  0.88519\\
0.15  0.83532\\
0.2  0.78795\\
0.25  0.74229\\
0.3  0.69779\\
0.35  0.65407\\
0.4  0.61079\\
0.45  0.56766\\
0.5  0.52441\\
0.55  0.48074\\
0.6  0.43635\\
0.65  0.3909\\
0.7  0.34402\\
0.75  0.29524\\
0.8  0.24404\\
0.85  0.18977\\
0.9  0.13166\\
0.95  0.068773\\
1  0\\
};
\addlegendentry{Bound $\frac{1}{6}R_{\Theta}(\newD_{2},6)(\cA_{p,0})$}

\addplot [color=mycolor3, line width=1.5pt, dotted]
  table[row sep=crcr]{%
0	1.0000122376519\\
0.025	0.981864556738505\\
0.05	0.963487028379919\\
0.1	0.926003073430793\\
0.2	0.848000492348132\\
0.3	0.765541627600207\\
0.4	0.678077466946298\\
0.5	0.584967868575473\\
0.6	0.48543167442716\\
0.7	0.378517121200577\\
0.8	0.263035706749808\\
0.9	0.137512932339375\\
0.95	0.0703955823322711\\
1	0\\
};
\addlegendentry{Bound $R_{\Theta}(\widehat{D}_{2},1)(\cA_{p,0})$}

\end{axis}
\end{tikzpicture}%
 	\caption{Comparison between two strong converse bounds $R_{\Theta}(\widehat{D}_2,1)$ and $\frac{1}{6}R_{\Theta}(\newD_{2},6)$ on for two-way/PPT assisted quantum capacity for the qubit amplitude damping channel $\cA_{p,0}$ for $p \in [0,1]$}  
 	\label{fig:quantum_capacity}
 \end{figure}

\section*{Acknowledgments}
We would like to thank Hamza Fawzi for useful discussions. HT would like to thank Sven Carel Polak for helpful discussions. This project has received funding from the European Research Council (ERC Grant Agreement No. 851716).  The research of HT is supported by the LABEX MILYON (ANR-10-LABX-0070) of Universit\'e de Lyon, within the program ``Investissements d'Avenir" (ANR-11-IDEX-0007) operated by the French National Research Agency (ANR).


\newpage
\phantomsection
\addcontentsline{toc}{section}{Bibliography}
\bibliographystyle{alphaurl}
\bibliography{bibliofile}
\newpage

\appendix

\end{document}